\documentclass[aps,preprint]{revtex4-2}

\usepackage[utf8]{inputenc}
\usepackage[T1]{fontenc}
\usepackage{amsmath}
\usepackage{amsthm}
\usepackage{amssymb}
\usepackage{hyperref}
\usepackage{dsfont}
\usepackage{graphicx}
\usepackage{tikz-cd}

\hypersetup{colorlinks=true, linkcolor=blue, citecolor=blue}

\newcommand{\fxn}[3]{$ #1 \colon #2 \rightarrow #3 $}

\newcommand{\Real}{\mathbb{R}}

\newcommand{\Comp}{\mathbb{C}}
\newcommand{\seq}{\subseteq}

\newcommand{\pspace}{\Real^3 \setminus \{0\}}
\newcommand{\LittleGroup}[1]{H_{#1}}
\newcommand{\Lightcone}{\mathcal{L}_+}
\newcommand{\xSpace}[1]{\boldsymbol{\mathcal{#1}}}
\newcommand{\kOmegaSpace}[1]{\tilde{\boldsymbol{#1}}}

\newcommand{\poincare}{\mathrm{ISO}^+(3,1)}

\newtheorem{theorem}{Theorem}
\newtheorem{lemma}[theorem]{Lemma}
\newtheorem{corollary}[theorem]{Corollary}
\newtheorem{proposition}[theorem]{Proposition}
\newtheorem{definition}[theorem]{Definition}

\begin{document}

% Author information
\author{Eric Palmerduca}
\email{ep11@princeton.edu}
\affiliation{Department of Astrophysical Sciences, Princeton University, Princeton, New Jersey 08544}
\affiliation{Plasma Physics Laboratory, Princeton University, Princeton, NJ 08543,
U.S.A}

\author{Hong Qin}
\email{hongqin@princeton.edu}
\affiliation{Department of Astrophysical Sciences, Princeton University, Princeton, New Jersey 08544}
\affiliation{Plasma Physics Laboratory, Princeton University, Princeton, NJ 08543,
U.S.A}

% Date/Title
\title{Photon topology}
\date{\today}

%%%%%%%%%%%%%%%%%% Abstract %%%%%%%%%%%%%%
\begin{abstract}
The topology of photons in vacuum is interesting because there are no photons with $\boldsymbol{k}=0$, creating a hole in momentum space. We show that while the set of all photons forms a trivial vector bundle $\gamma$ over this momentum space, the $R$- and $L$-photons form topologically nontrivial subbundles $\gamma_\pm$ with first Chern numbers $\mp2$. In contrast, $\gamma$ has no linearly polarized subbundles,  and there is no Chern number associated with linear polarizations. It is a known difficulty that the standard version of Wigner's little group method produces singular representations of the Poincar\'{e} group for massless particles. By considering representations of the Poincar\'{e} group on vector bundles we obtain a version of Wigner's little group method for massless particles which avoids these singularities. We show that any massless bundle representation of the Poincar\'{e} group can be canonically decomposed into irreducible bundle representations labeled by helicity, which in turn can be associated to smooth irreducible Hilbert space representations. This proves that the $R$- and $L$-photons are globally well-defined as particles and that the photon wave function can be uniquely split into $R$- and $L$-components. This formalism offers a method of quantizing the EM field without invoking discontinuous polarization vectors as in the traditional scheme. We also demonstrate that the spin-Chern number of photons is not a purely topological quantity. Lastly, there has been an extended debate on whether photon angular momentum can be split into spin and orbital parts. Our work explains the precise issues that prevent this splitting. Photons do not admit a spin operator; instead, the angular momentum associated with photons' internal degree of freedom is described by a helicity-induced subalgebra corresponding to the translational symmetry of $\gamma$.
\end{abstract}

\maketitle

%%%%%%%%%% Section 1: Introduction %%%%%%%%%%%%%%%%%%
\section{Introduction}

Following the breakthrough discoveries of the quantum hall effect, the Berry connection, and the existence of topological insulators in condensed matters with periodic lattice structures, topology has played an increasingly important role in physics. It is now understood that there is inherent topology in wave physics, and researchers have studied the topological properties of waves in various continuous media such as fluids \cite{Souslov2017,Delplace2017,Tauber2019,Souslov2019,Perrot2019,Faure2022} and plasmas \cite{Yang2016,Gao2016,Parker2020a,Fu2021,Fu2022,Qin2023}. Here, we examine the topology of a simpler system, that of Maxwell's equations in vacuum. That this system can exhibit topologically nontrivial behavior is due to the fact that there are no photons with wave vector $\boldsymbol{k}=0$, creating a hole in the momentum space $M$. We will show that the collection of all wave solutions in Fourier space form a vector bundle $\gamma$ over $M$, which we call the photon bundle.

Using this framework, we give an extensive study of the topological properties of photons by examining the topological properties of $\gamma$ and its subbundles. We give two proofs that the photon bundle $\gamma$ is topologically trivial. The first is an abstract proof using Chern numbers and results from low-dimensional topology. The second proof involves an explicit construction of two linearly-independent, globally defined polarization vector fields via the clutching construction from algebraic topology \cite{HatcherVBKT}. These are complex polarization vectors perpendicular to the wave vector $\boldsymbol{k}$, and are mathematically equivalent to sections of the complexification of $TS^2$, the tangent bundle of the 2-sphere. It has been claimed in some quantum field theory textbooks \cite{Tong2006, Woit2017} that such global polarization vectors do not exist as they would violate the hairy ball theorem, which states that the tangent bundle of the sphere $TS^2$ is trivial (\cite{Frankel2011}, Cor. 16.10). However, as we show, the hairy ball theorem only applies to $TS^2$, not its complexification, which resolves the apparent paradox.

While the total photon bundle is trivial, we show that the right ($R$)- and left ($L$)-circularly polarized photons form globally well-defined nontrivial subbundles with first Chern numbers $\mp2$. These Chern numbers were also obtained by Bliokh \emph{et al.} in their study of the quantum spin Hall effect \cite{Bliokh2015}. They calculated the Chern numbers via the Berry connection, which is a geometric quantity, but the concepts of  vector bundle and subbundle were not defined or discussed. We emphasize that it is crucial to establish the topological structures of vector bundles before calculating their Chern numbers. Due to the Chern-Weil homomorphism, a deep result from differential geometry,  the Berry connection can be used to calculate the topologically invariant Chern numbers \cite{Tu2017differential}. Generally though, quantities derived from the Berry connection will depend on the geometry as well as topology of the underlying vector bundle. Indeed, while Bliokh \emph{et al.} \cite{Bliokh2015} also define the spin Chern number of light, we show that this quantity is not a true topological invariant like the ordinary Chern numbers.

To further illustrate the importance of vector bundle structures, we investigate the familiar linearly polarized photons. As circularly polarized photons have nontrivial topology characterized by nonvanishing Chern numbers, one might guess the same is true of linear polarizations. We show that this question is actually ill-defined as there are no linearly polarized subbundles of $\gamma$. This implies that linear polarization is only a locally defined concept in momentum space, and it is meaningless to discuss its topology or Chern number. 

In addition to these purely topological results, we study the interplay of this topology with the geometry of light, that is, with Poincar\'{e} symmetry. Poincar\'{e} symmetry plays a fundamental role in the theory of photons as elementary particles, both massive and massless,  are defined as unitary irreducible representations of the Poincar\'{e} group. Conventionally, these are assumed to be vector space (in particular, Hilbert space) representations, and are classified via Wigner's little group method \cite{Wigner1939,Weinberg1995}. This theory works well for massive particles, however, a topological singularity occurs when the mass is taken to zero. In this limit, the momentum space jumps from being a topologically trivial mass hyperboloid to the nontrivial lightcone. This results in non-smooth massless representations generated by the standard little group method.  This difficulty has been noted by others \cite{Flato1983,Dragon2022}. Here, we present a solution to this problem by considering vector bundle representations rather than vector space representations. Such representations are known to mathematicians under the name of equivariant vector bundles, but they are rarely used in physics. As such, we will develop the necessary theory here. Using these vector bundle representations, we present a version of the little group theory for vector bundle representations that does not encounter discontinuities. We prove that massless irreducible unitary vector bundle representations of the Poincar\'{e} group induce corresponding unitary irreducible representations on the Hilbert space of bundle sections, and thus correspond to particles by the conventional definition. We apply this method in the case of photons to show that $R$- and $L$-polarized photons are massless unitary irreducible vector bundle representation of the Poincar\'{e} group labeled by helicity, and thus correspond to globally well-defined particles. 

This equivariant bundle formalism also lends itself to other problems in the theory of photons. The standard quantization of the electromagnetic field in the Coulomb gauge involves expanding the field in global $R$- and $L$-polarization bases which cannot be continuous \cite{Tong2006,Folland2008,Woit2017}. While continuous global $R$- and $L$-polarization vectors do not exist, the $R$- and $L$-photon bundles are well-defined. Using projections onto these bundles, we show that the electromagnetic field can be quantized to obtain the usual QED creation and annihilation operators without using discontinuous bases. 

Lastly, there has been a long debate as to whether photon angular momentum can be split into spin and orbital parts \cite{Akhiezer1965, VanEnk1994,Bliokh2010, Bialynicki-Birula2011, Leader2013, Bliokh2015, Leader2016, Leader2019}. It is known that the most commonly proposed ``spin and orbital angular momentum operators'' satisfy peculiar commutation relations, different from the $\mathfrak{so}(3)$ commutation relations expected of angular momentum operators \cite{VanEnk1994, Bliokh2010, Leader2019}. Our theoretical study of the massless vector bundle representation confirms that photons do not admit a spin operator. We show that the proposed ``spin angular momentum operator'' is neither spin nor angular momentum, but a three-dimensional commuting subalgebra corresponding to the $\mathbb{R}^3$ translational symmetry of the photon bundle characterizing the photons' internal degree of freedom.

This paper is organized as follows. In section \ref{sec:TopologyPhotonBundle}, we construct the total photon bundle $\gamma$ and give two proofs that it is trivial. The second proof is used to explicitly construct two linearly independent global polarization vector fields. In section \ref{sec:TopologyRL}, we show that $R$- and $L$-photons form well-defined subbundles of $\gamma$ that are topologically nontrivial. We also show that no analogous linearly polarized subbundles exist. In section \ref{sec:BundleReps}, we define vector bundle representations and develop a modified little group construction. These are used to classify photons as unitary irreducible vector bundle representations of the Poincar\'{e} group. We prove that massless irreducible vector bundle representations of the Poincar\'{e} group naturally induce unitary irreducible Hilbert space representations, and thus correspond to particle via the conventional definition. In section \ref{sec:applications}, we apply our theory to three problems in the theory of photons: the geometric nature of the spin Chern number, the quantization of the electromagnetic field, and the decomposition of photon angular momentum into spin and orbital parts.

%%%%%%%%% Section 2: Topology of Photon Bundle %%%%%%%
\section{Topology of the photon bundle}\label{sec:TopologyPhotonBundle}
Consider photons in the vacuum, defined as the eigenmodes of Maxwell equations,
\begin{gather}
    \partial_t \boldsymbol{\mathcal{E}} = \nabla \times \boldsymbol{\mathcal{B}} \\
    \partial_t \boldsymbol{\mathcal{B}} = -\nabla \times \boldsymbol{\mathcal{E}} \\
    \nabla \cdot \boldsymbol{\mathcal{E}} = 0 \\
    \nabla \cdot \boldsymbol{\mathcal{B}} = 0,
\end{gather}
where $t$ has been normalized by $1/c$ and has the unit of length. Via the Fourier transform, the fields can be expressed as
\begin{gather}
\xSpace{E}(x) = \int \frac{d^4k}{(2\pi)^4}e^{-i k_\mu x^\mu}\kOmegaSpace{E}(k) \\
\xSpace{B}(x) = \int \frac{d^4k}{(2\pi)^4}e^{-i k_\mu x^\mu}\kOmegaSpace{B}(k),
\end{gather}
where we have used the four-vector notation $x^\mu = (t,\boldsymbol{x})$ and $k^\mu = (\omega,\boldsymbol{k})$, along with the flat spacetime metric $\eta_{\mu \nu}=\text{diag}(-1,+1,+1,+1)$. As Maxwell's equations have a real discrete spectrum, the momentum space fields can be written as
\begin{gather}
    \kOmegaSpace{E}(k) \doteq 2\pi \sum_j \delta\big(\omega - \omega_j(\boldsymbol{k})\big) \boldsymbol{E}_j(\boldsymbol{k}) \label{eq:E_Fourier}\\
    \kOmegaSpace{B}(k) \doteq 2\pi \sum_j \delta\big(\omega - \omega_j(\boldsymbol{k})\big) \boldsymbol{B}_j(\boldsymbol{k})
\end{gather}
where the $\omega_j$ and $\boldsymbol{E}_j$ are solutions of the eigenvalue problem
\begin{gather}
    H\left(\begin{array}{c}
    \boldsymbol{E}_j\\
    \boldsymbol{B}_j
    \end{array}\right)=\omega_j \left(\begin{array}{c}
    \boldsymbol{E}_j\\
    \boldsymbol{B}_j
    \end{array}\right) \\
    H=\left(\begin{array}{cc}
    \boldsymbol{0} & \boldsymbol{-k}\times\\
    \boldsymbol{k}\times & \boldsymbol{0}
    \end{array}\right)
\end{gather}
subject to the constraints
\begin{align}
    \boldsymbol{k} \cdot \boldsymbol{E}_j &= 0 \label{eq:k_dot_E}\\
    \boldsymbol{k} \cdot \boldsymbol{B}_j &= 0 \label{eq:k_dot_B}.
\end{align}

The spectrum of $H$ has 6 eigenmodes (photons) divided into three two-fold degenerate groups: 

1. Two transverse photons with $\omega_{1}=|\boldsymbol{k}| \equiv\sqrt{k_{x}^{2}+k_{y}^{2}+k_{z}^{2}}$ ; 

2. Two transverse anti-photons with $\omega_{-1}=-|\boldsymbol{k}|$; 

3. Two longitudinal photons with $\omega_{0}=0$. 

The longitudinal photons are forbidden by the transversality conditions (\ref{eq:k_dot_E}) and (\ref{eq:k_dot_B}), although they can appear as virtual photons in QED \cite{Gupta1950}. Furthermore, since the
spacetime fields $\xSpace{E}$ and $\xSpace{B}$ are real, the momentum space fields must satisfy $\boldsymbol{E}_{-1}(\boldsymbol{k}) = \boldsymbol{E}_1^*(\boldsymbol{-k})$ and $\boldsymbol{B}_{-1}(\boldsymbol{k}) = \boldsymbol{B}_1^*(\boldsymbol{-k})$, and thus the transverse anti-photons are not independent of the transverse photons. Therefore, we need only analyze the transverse photons, which will henceforth be referred to just as photons with $\boldsymbol{E}\doteq \boldsymbol{E}_1$ and $\boldsymbol{B}\doteq \boldsymbol{B}_1$. The condition $\omega = |\boldsymbol{k}|$ says that photons reside on the forward lightcone $\mathcal{L}_+$ in 4-momentum space. Note that the origin is not contained in $\mathcal{L}_+$. Physically, this corresponds to the fact that $k^\mu=0$ modes are stationary and do not propagate at the speed of light, and therefore are not photons. Mathematically, the origin must be removed in order for the forward lightcone to be a regular submanifold of Minkowski space. $\mathcal{L}_+$ is a 3 dimensional manifold homeomorphic to $\pspace$ with coordinates given by $\boldsymbol{k} = \{k_x,k_y,k_z \}\neq 0$. At each $\boldsymbol{k}$, the two-fold degenerate photon eigenvectors form a $\mathbb{C}^{2}$ vector space consisting of vectors $\boldsymbol{E}$ satisfying $\boldsymbol{E} \cdot \boldsymbol{k} = 0 $. Note that $\boldsymbol{B}=  \hat{\boldsymbol{k}} \times \boldsymbol{E}$, so these vector spaces can be described purely in terms of the electric field. Since $H$ is smooth and the dimension of the eigenspaces are constant, the eigenspaces for all values of $\boldsymbol{k}$ fit together to form a rank-2 complex vector bundle $\pi:\gamma \rightarrow \mathcal{L}_+$ over the lightcone, which we term the photon bundle. Here, $\pi$ is the projection of the bundle $\gamma$ onto the base manifold and $\gamma(\boldsymbol{k})$ will denote the fiber at $\boldsymbol{k}$. $\gamma$ is a Hermitian bundle with respect to the Hermitian product inherited by embedding $\boldsymbol{E}$ in $\Comp^3$. That is, for two vectors vectors $\boldsymbol{E}_a, \boldsymbol{E}_b \in \gamma(\boldsymbol{k})$,
\begin{equation}\label{eq:Hermitian_product_E}
    \langle \boldsymbol{E}_a , \boldsymbol{E}_b \rangle = \boldsymbol{E}_a^* \cdot \boldsymbol{E}_b.
\end{equation}

We are concerned here with the topology of photons, which is equivalent to the topology of the bundle $\gamma$. If the base manifold $\mathcal{L}_+$ were contractible, then $\gamma$ would necessarily be trivial (\cite{HatcherVBKT}, Cor. 1.8). Recall that a rank-$r$ vector bundle over $M$ is trivial if it is isomorphic to the trivial bundle $M \times \mathbb{C}^r$, or equivalently, if there exist $r$ nonvanishing sections of the bundle which are pointwise linearly independent. For example, for massive particles, the momenta $k^\mu$ form a mass hyperboloid $(k^0)^2 - \boldsymbol{k}\cdot \boldsymbol{k} = m^2$ which is homeomorphic to $\mathbb{R}^3$ and thus contractible. As a result, massive particles in vacuum do not allow for nontrivial topology. However, for massless photons, the base manifold $\mathcal{L}_+ \cong \pspace$ is not topologically trivial in the sense of homotopy, so it is not obvious if $\gamma$ is a trivial bundle. The following theorem is one of the main results of this paper:

\begin{theorem}\label{thm:gamma_trivial}
The photon bundle $\gamma \rightarrow \pspace$ is trivial.
\end{theorem}

We will give two proofs of this result in the next two sections. The first is purely abstract using the machinery of Chern classes and results from low-dimensional topology. The second proof utilizes clutching functions which has the advantage of giving an explicit construction of linearly independent sections of $\gamma$.

We give some remarks about Theorem \ref{thm:gamma_trivial} before presenting the proofs. It is often more convenient to work with $\gamma|_{S^2}$, the bundle obtained by restricting the base manifold of $\gamma$ to the unit sphere $S^2 \in \pspace$, rather than with $\gamma$ itself. Since $\pspace$ deformation retracts onto $S^2$, the bundles $\gamma$ and $\gamma|_{S^2}$ are essentially topologically equivalent. In particular, the rank-$r$ vector bundles over $\pspace$ and $S^2$ are in bijective correspondence (\cite{HatcherVBKT}, Cor 1.8). Bundles over $\pspace$ map to bundles over $S^2$ by restricting the base manifold. If we denote by $r:\pspace \rightarrow S^2$ the projection $\boldsymbol{k} \rightarrow \hat{\boldsymbol{k}}$, then the inverse mapping sends a bundle $\xi$ over $S^2$ to the pullback bundle $r^* \xi$ over $\pspace$.

There are two main advantages to working with $S^2$, the first being the immediate observation that since the electric field is transverse to $\hat{\boldsymbol{k}}$, $\gamma|_{S^2}$ is isomorphic to $T^\mathbb{C} S^2 \cong TS^2 \times_\mathbb{R} \mathbb{C}$, the complexification of the tangent bundle $TS^2$ of the sphere. Thus,
\begin{lemma}\label{lm:restriction_to_sphere}
    The photon bundle $\gamma$ is trivial if and only if $\gamma|_{S^2} \cong T^\mathbb{C} S^2$ is trivial.
\end{lemma}
This lemma allows the use of specialized techniques that apply to bundles over spheres such as the clutching construction. Unlike $\pspace$, $S^2$ also has the advantage of being even-dimensional and thus possessing Chern numbers in addition to Chern classes. This is useful because Chern numbers are simply integers rather than cohomology classes.

By Lemma \ref{lm:restriction_to_sphere}, the question of whether or not $\gamma$ is trivial can be rephrased as asking if it is possible to smoothly choose polarization vectors $\boldsymbol{E}$ for each direction $\hat{\boldsymbol{k}}$. It can also be interpreted as asking is it possible to have photons traveling in all directions simultaneously. This issue shows up in many treatments of QED when quantizing the electromagnetic field \cite{Tong2006, Folland2008, Woit2017}. Just prior to quantization, the Coulomb gauge vector potential is usually expressed in the form
\begin{equation}\label{eq:vector_potential}
    \boldsymbol{\mathcal{A}}(x) = \int \sum_{j=1}^2 \frac{1}{\sqrt{2|\boldsymbol{k}|}} [e^{-ik_\mu x^\mu}\boldsymbol{\epsilon}_j (\boldsymbol{k}) a_j(\boldsymbol{k}) +e^{i k_\mu x^\mu}\boldsymbol{\epsilon}^*_j (\boldsymbol{k})a^*_j(\boldsymbol{k})]\frac{d^3\boldsymbol{k}}{(2\pi)^3}.
\end{equation}
Here, $\boldsymbol{\epsilon}_j(\boldsymbol{k})$ are polarization unit vectors which are supposed to satisfy the transversality condition $\boldsymbol{\epsilon}_j(\boldsymbol{k})\cdot \boldsymbol{k}=0$ imposed by the Coulomb gauge condition $\nabla \cdot \boldsymbol{\mathcal{A}} = 0$. This gauge condition imposes the same transversality constraint on the $\boldsymbol{\epsilon}_j$ that Gauss's law imposes on $\boldsymbol{E}$, and we see that $a_j(\boldsymbol{k})\boldsymbol{\epsilon}_j(\boldsymbol{k})$ and its conjugate can be considered as sections of $\gamma$. Despite using this expansion, Tong \cite{Tong2006} and Woit \cite{Woit2017} claim that no continuous functions $\boldsymbol{\epsilon}_j(\boldsymbol{k})$ actually exist. Indeed, this is true if one requires the $\boldsymbol{\epsilon}_j$ to be real, that is, if one considers only linear polarizations. In particular, if the $\boldsymbol{\epsilon}_j$ are real, then they are sections of $TS^2$ when restricted to $S^2$. By the hairy ball theorem, there are no continuous nonvanishing sections of $TS^2$. Despite this fact, Staruszkiewicz attempted to construct real, global polarization vectors \cite{Staruszkiewicz1973a}. These vector fields are not actually global though as they are only defined on stereographic coordinates and thus exclude one of the poles of the momentum sphere. Since we are working in momentum space, it is natural to work with complex polarization vectors, not just linear polarizations. R- and L-circular polarizations play a particularly important role due to their invariance under Lorentz transformations. The question can then be asked, is it possible to comb complex hairs on a ball flat? By explicit construction, it is indeed possible.
\begin{theorem}
There exists at least one nonvanishing section of the bundles $\gamma$ and $\gamma|_{S^2} \cong T^\Comp S^2$.
\end{theorem}
\begin{proof}
    Consider the sections of $\gamma$:
    \begin{align}
        \boldsymbol{E}_1(\boldsymbol{k}) = (0,k_z,-k_y), \\ \boldsymbol{E}_2(\boldsymbol{k}) = (-k_z,0,k_x).
    \end{align}
    Then
    \begin{equation}
        \boldsymbol{u} = \boldsymbol{E}_1 + i\boldsymbol{E}_2
    \end{equation}
    is a nonvanishing section of $\gamma$ since
    \begin{equation}
        |\boldsymbol{u}| = \sqrt{k_x^2 + k_y^2 + 2k_z^2} > 0.
    \end{equation}
    By restricting $\boldsymbol{u}$ to $S^2$ one obtains a nonvanishing section of $\gamma|_{S^2} \cong T^\Comp S^2$.
\end{proof}
Thus, it is possible to consistently choose at least one of the $\boldsymbol{\epsilon}_j$. However, to prove Theorem \ref{thm:gamma_trivial}, one would need to construct a second linearly independent polarization, which is a much harder task. For example, while one might guess that $\boldsymbol{u}' = \boldsymbol{E}_1 -i \boldsymbol{E}_2$ may be such a vector field, $\boldsymbol{u}$ and $\boldsymbol{u}'$ are linearly dependent at any $\boldsymbol{k}$ where $\boldsymbol{E}_1$ or $\boldsymbol{E}_2$ vanishes. We note that Folland \cite{Folland2008} states without proof that two such complex polarizations do exist. However, to our knowledge, there is neither an explicit construction of such polarizations nor a proof of their existence in the literature. We fill in these gaps in the remainder of this section.

%%%%%%%%%%%%% Triviality Proof 1 %%%%%%%%%%%%%%
\subsection{Proof of Theorem \ref{thm:gamma_trivial} via Chern numbers}\label{subsec:proof_1}
The first proof of Theorem \ref{thm:gamma_trivial} relies on Chern classes, the characteristic classes of complex vector bundles. An advantage of this approach is that Chern classes are commonly used in physical applications of topology, such as in the study of topological insulators and the topological properties of waves.

By Lemma \ref{lm:restriction_to_sphere}, it is sufficient to show $\gamma|_{S_2}$ is trivial. We begin by proving that all of the Chern classes $c_j$ of $\gamma|_{S_2}$ are 0.

\begin{definition}[Conjugate vector bundle]
Let \fxn{\pi}{E}{M} be an $n$-dimensional complex vector bundle. On each fiber $E_p$, $p\in M$, define multiplication of a complex number $a+ib \in \Comp$ with a vector $v \in E_p$ by 
\begin{equation}
    (a+ib)v = av - ibv
\end{equation}
which forms a new complex vector space $\overline E_p$. Then
\begin{equation}
    \overline E = \bigcup \overline E_p
\end{equation}
is defined to be the conjugate bundle of $E$.
\end{definition}
In general, complex vector bundles are not isomorphic to their conjugates. However,
\begin{lemma}\label{lm:conjugate_bundle_lemma}
The complexification $E^\Comp = E^\Real \otimes_\Real \Comp$ of a real vector bundle $E^\Real$ is isomorphic to its own conjugate bundle $\overline{E^\Comp} = \overline{E^\Real \otimes_\Real \Comp}$.
\end{lemma}
See Ref. \cite{Milnor1974}, Lemma 15.1 for a proof. This helps us calculate Chern classes by the following lemma.

\begin{lemma}\label{lm:complexification_lemma}
Let $E$ be a complex vector bundle. The Chern classes $c_j$ of $E$ and $\overline{E}$ are related by
\begin{equation}
    c_j(\overline{E}) = (-1)^jc_j(E).
\end{equation}
\end{lemma}
See Ref. \cite{Milnor1974}, Lemma 14.9 for a proof.

\begin{theorem}\label{thm:Chern0}
The Chern classes and Chern numbers of $\gamma|_{S^2}$ are 0.
\end{theorem}
\begin{proof}
Since $\gamma |_{S^2}$ is a vector bundle over the $2$-dimensional manifold $S^2$ and $c_j$ is represented by a $2j$-form on $S^2$, $c_j(\gamma |_{S^2}) = 0$ for all $j>1$. By Lemmas \ref{lm:conjugate_bundle_lemma} and \ref{lm:complexification_lemma}, 
\begin{equation}
\begin{split}
    c_1(\gamma |_{S^2}) &= c_1(T^\Comp S^2) = c_1(\overline{T^\Comp S^2}) \\
    &= -c_1(T^\Comp S^2) = -c_1(\gamma |_{S^2}),
\end{split}    
\end{equation}
so $c_1(\gamma |_{S^2}) = 0$. The Chern numbers are integrals of products of Chern classes, so it follows that the Chern numbers are also 0.
\end{proof}

While the Chern classes do not completely classify general complex vector bundles, this is true for certain low-dimensional, low-rank vector bundles. Here, we show that the vanishing of the Chern classes is sufficient to prove the triviality of the bundle $\gamma$  by using results about symplectic vector bundles, relying heavily on results from Ref. \cite{McDuff2017}. A symplectic vector bundle is a real vector bundle $\pi:E\rightarrow M$ equipped with a symplectic 2-form $\omega$. Complex vector bundles and symplectic vector bundles have the same characteristic classes, namely the Chern classes \cite{McDuff2017}. We have the following classification theorem:
\begin{theorem}\label{thm:symp}
Two symplectic vector bundles $E$ and $E'$ over a closed (\emph{i.e.}, compact without boundary) oriented $2$-manifold are isomorphic if and only if they have the same rank and the same first Chern number.
\end{theorem}
\begin{proof}
Ref. \cite{McDuff2017}, Theorem 2.7.1.
\end{proof}
The relationship between symplectic and complex vector bundles is given by the following theorem:
\begin{theorem}\label{thm:equiv}
For $j=1,2$ let $(E_j,\omega_j)$ be a symplectic vector bundle over a manifold $M$ and let $J_j$ be a complex structure on $E_j$ satisfying $\omega_j(v,J_j v)>0$ for all $v\in E_j$. Then the symplectic vector bundles $(E_1,\omega_1)$ and $(E_2, \omega_2)$ are isomorphic if and only if the complex vector bundles $(E_1,J_1)$ and $(E_2,J_2)$ are isomorphic. 
\end{theorem}
\begin{proof}
Ref. \cite{McDuff2017}, Theorem 2.6.3.
\end{proof}
In the context of $\gamma|_{S^2}$, the complex structures are just multiplication by $i$. We combine the preceding three theorems to prove the main result.
\begin{theorem}
$\gamma|_{S^2}$ is trivial.
\end{theorem}

\begin{proof}
For $j=1,2$, let $w_j = \alpha_j + i\beta_j \in \gamma|_{S^2}(\boldsymbol{k})$ for some $\boldsymbol{k}$, where $\alpha_j,\beta_j \in \Real^3$. $\gamma|_{S_2}$ has the Hermitian structure
\begin{align}\label{eq:Hermitian}
    \langle w_1, w_2 \rangle = w_1^* \cdot w_2 = (\alpha_1 \cdot \alpha_2 + \beta_1 \cdot \beta_2) + i(\alpha_1 \cdot \beta_2 - \beta_1 \cdot \alpha_2)
\end{align}
Define the bilinear form $\omega$ on $\gamma|_{S^2}$ by
\begin{equation}
    \omega(w_1,w_2) = \alpha_1 \cdot \beta_2 - \beta_1 \cdot \alpha_2.
\end{equation}
Eq.\,(\ref{eq:Hermitian}) can be written as
\begin{equation}
    \langle w_1, w_2 \rangle = \omega(w_1,i w_2) + i\omega(w_1,w_2).
\end{equation}
Since $\omega$ is nondegenerate and skew-symmetric, it defines a symplectic form on each fiber $\gamma|_{S_2}(\boldsymbol{k})$. Since the Hermitian form is smooth in $\boldsymbol{k}$, and since $\omega$ is the imaginary part of the Hermitian form, $\omega$ is also smooth in $\boldsymbol{k}$. Thus, $(\gamma|_{S^2},\omega)$ is a symplectic vector bundle. Furthermore,
\begin{equation}
    \omega(w_1,i w_1) = |w_1|^2 > 0.
\end{equation}
The definitions of the Chern classes for complex vector bundles and symplectic vector bundles agree \cite{McDuff2017},  so by Theorems \ref{thm:Chern0} and \ref{thm:symp}, $(\gamma|_{S^2},\omega)$ is trivial as a symplectic bundle. By Theorem \ref{thm:equiv}, $\gamma|_{S^2}$ is trivial as a complex vector bundle.
\end{proof}

%%%%%%%%%%%%%%%% Triviality Proof 2 %%%%%%%%%%%%%%%%%%%
\subsection{Proof of Theorem \ref{thm:gamma_trivial} via the clutching construction}
Here we give a second proof that $\gamma$ is trivial via the clutching construction \cite{HatcherVBKT}. Again, we proceed by showing $T^\Comp S^2$ is trivial. One advantage of this method is that it can be adapted to explicitly furnish two independent vector fields on $T^\Comp S^2$. 

The clutching construction is a method for determining the isomorphism type of a real or complex vector bundle over $S^n$; here, we only consider the case of rank-2 complex bundles over $S^2$. Consider a rank-2 complex vector bundle $\pi: E \rightarrow S^2$, and decompose $S^2$ into its upper and lower hemispheres, $S^2 = D^2_+ \cup D^2_-$. Since each hemisphere is contractible, the restricted bundles $E|_{D^2_+}$ and $E|_{D^2_-}$ are trivial. Fix a trivialization for each of these bundles, and let $[v_{1\pm},v_{2\pm}]$ be an orthonormal frame for $D^2_\pm$. These frames will generally not agree on the equator $D^2_+\cap D^2_-=S^1$ where the hemispheres overlap, and there is a function \fxn{f}{S^1}{\mathrm{U}(2)} which rotates one frame into the other on $S^1$:
\begin{equation}
    [v_{1+},v_{2+}](\phi) = f(\phi)[v_{1-},v_{2-}](\phi).
\end{equation}
$f$ is referred to as a clutching function for E. 

Conversely, given an arbitrary function \fxn{f}{S^1}{\mathrm{U}(2)}, one can construct a rank-2 complex vector bundle $E_f$ (see Ref. \cite{HatcherVBKT}, p. 22). For smooth manifolds $M$ and $N$, let $[M,N]$ denote the set of homotopy classes of maps from $M$ to $N$, and let $\text{Vect}_\Comp^r(S^2)$ be the set of isomorphism classes of rank-$r$ complex vector bundles over $S^2$. The following fundamental theorem says that the homotopy type of a clutching function determines the isomorphism class of a vector bundle.

\begin{theorem}\label{thm:clutching_theorem}
The map $\Phi:[S^{1},\mathrm{U}(2)]\rightarrow \mathrm{Vect}_\Comp^2(S^2)$ which sends a clutching function $f$ to the vector bundle $E_f$ is a bijection.
\end{theorem}
\begin{proof}
This is a special case of proposition 1.11 in \cite{HatcherVBKT}.
\end{proof}

$[S^1,\mathrm{U}(2)] = \pi_1(\mathrm{U}(2))$, where $\pi_1$ denotes the fundamental group, and $\pi_1(\mathrm{U}(2))\cong \mathbb{Z}$, so the complex vector bundles over $S^2$ can be labeled by the integers, with the trivial bundle corresponding to $0$. It is well known that $\mathrm{SU}(2)$ is simply connected (for example, by the fact that it is diffeomorphic to $S^3$), so we get the following corollary.

\begin{corollary}\label{cor:cor_SU2}
If a rank-2 vector bundle \fxn{\pi}{E}{S^2} has a clutching function $f$ which factors through $\mathrm{SU}(2)$, that is, if $f:S^1\rightarrow \mathrm{SU}(2) \seq \mathrm{U}(2)$, then $E$ is trivial.
\end{corollary}
\begin{proof}
Since $\text{SU}(2)$ is simply connected, $f$ is homotopic to the constant map $\mathds{1}:S^1\rightarrow \mathrm{U}(2)$ which sends every element of $S^1$ to the identity of $\mathrm{U}(2)$. $\mathds{1}$ is the clutching function for the trivial bundle, so the corollary follows by Theorem \ref{thm:clutching_theorem}.
\end{proof}

\begin{figure}
    \centering
    \includegraphics[width=8.6cm]{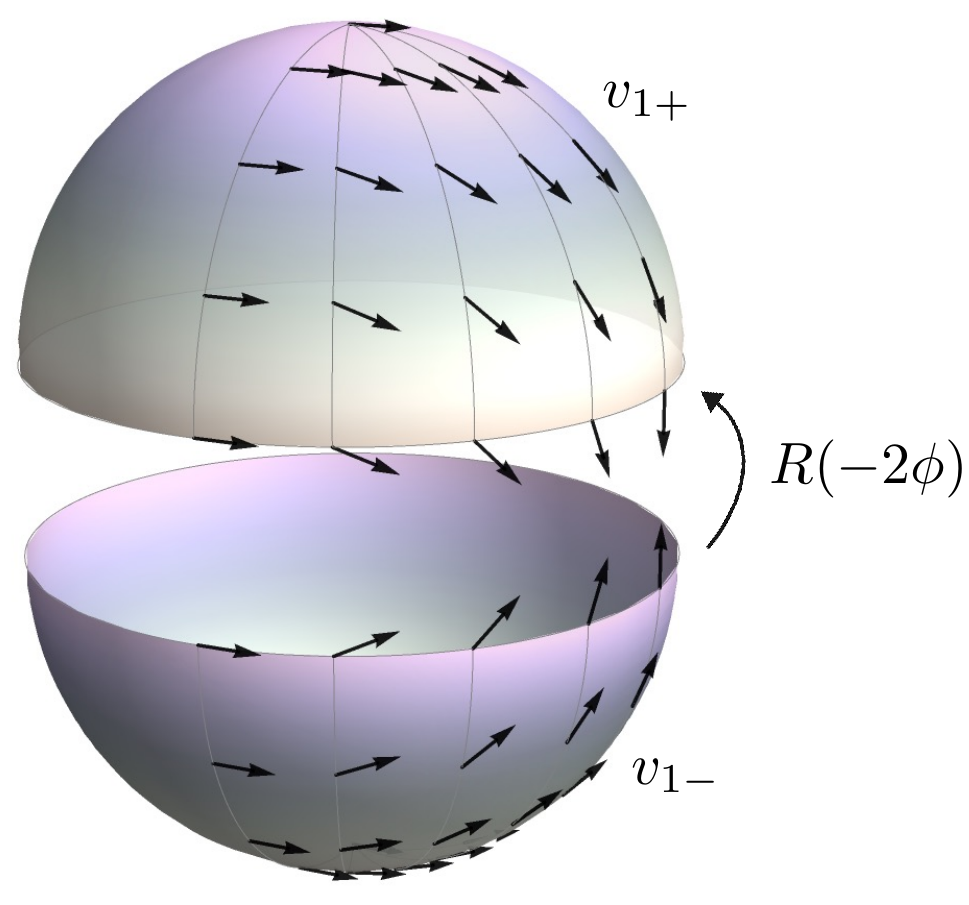}
    \caption{Vector fields $v_{1+}$ on $D^2_+$ and $v_{1-}$ on $D^2_-$. The clutching function relating them along the equator is a $-2\phi$ rotation.}
    \label{fig:clutching_sphere}
\end{figure}

We now use this corollary to prove $T^\Comp S^2$ is trivial. We construct a clutching function for $T^\Comp S^2$ by adapting a construction of a clutching function of $T^\Real S^2$ given by Hatcher (\cite{HatcherVBKT}, p. 22). We start by constructing orthonormal frames for the two hemispheres $D^2_\pm$ of $S^2$. Choose a vector field $v_{1+}$ by first picking some real tangent vector at the north pole and then extending it along each meridian so that it maintains a constant angle with the meridian (see Figure \ref{fig:clutching_sphere}). Let $v_{1-}$ be the vector field on the lower hemisphere $D^2_-$ obtained by reflecting $v_{1+}$ across the equatorial plane. We can produce a second real vector field $v_{2\pm}$ on each hemisphere by rotating $v_{1\pm}$ by 90 degrees counterclockwise when viewed from the exterior of the sphere. Then $[v_{1+},v_{2+}]$ and $[v_{1-},v_{2-}]$ are orthonormal trivializations of $D^2_+$ and $D^2_-$, respectively. Explicitly, by using Cartesian $(x,y,z)$ and spherical $(r,\theta,\phi)$ coordinates related in the standard way, and by choosing $v_{1+}$ to be $\hat y$ at the north pole, we have
\begin{align}
    v_{1+} &= \sin(\phi) \hat \theta + \cos (\phi) \hat \phi \label{eq:v1p}\\
    v_{2+} &= -\cos(\phi)\hat \theta + \sin(\phi) \hat \phi \\
    v_{1-} &= -\sin(\phi) \hat \theta + \cos(\phi) \hat\phi \\
    v_{2-} &= -\cos(\phi)\hat\theta - \sin(\phi)\hat \phi. \label{eq:v2m}
\end{align}
On the equator, the two frames are related by rotating the lower frame by $-2\phi$:
\begin{equation}
    [v_{1+},v_{2+}] = \begin{pmatrix}
    \cos 2\phi && \sin 2\phi \\
    -\sin 2\phi && \cos 2\phi
    \end{pmatrix}[v_{1-},v_{2-}] \doteq f(\phi)[v_{1-},v_{2-}].
\end{equation}
where \fxn{f}{S^1}{\mathrm{SU}(2)\seq \mathrm{U}(2)} is the clutching function and the matrix of $f$ and the column vectors are expressed in the $(\hat \theta, \hat \phi)$ basis. Since $f$ factors through $\mathrm{SU}(2)$, Corollary \ref{cor:cor_SU2} proves the following.
\begin{theorem}
$T^\Comp S^2$ is trivial.
\end{theorem}

%%%%%%%%%%%%%%%%% Explicit construction of frame %%%%%%%%%%%%%%%%
\subsection{Explicit construction of a frame on \texorpdfstring{$T^\Comp S^2$}{TCS2} and \texorpdfstring{$\gamma$}{g}}

Using methods related to the clutching construction, we can explicitly construct two independent vector fields on $T^\Comp S^2$, from which one can immediately obtain corresponding independent vector fields on $\gamma$. We do this in three steps. First, we find an explicit homotopy in $\mathrm{SU}(2)$ between the clutching function $f:S^1\rightarrow \mathrm{SU}(2)$ and the identity $\mathds{1}$. Second, we use this homotopy to extend the domain of $f$ to all of $D^2_-$. From this we will obtain a continuous, piecewise smooth frame for $T^\Comp S^2$. The last step is to modify the construction so that the frame is everywhere smooth.

%%%%%%%%%%%%%% \mathrm{SU}(2) Homotopy %%%%%%%%%%
\subsubsection{Homotopy from \texorpdfstring{$f$}{f} to the identity}
From the fact that $f$ factors through $\mathrm{SU}(2)$, we showed that $[f]=0$, meaning that there exists a homotopy between
\begin{equation*}
f(\phi) = 
    \begin{pmatrix}
    \cos 2\phi && \sin 2\phi \\
    -\sin 2\phi && \cos 2\phi
    \end{pmatrix}
    \;\;\; \text{ and } \;\;\;
    \mathds{1} = \begin{pmatrix}
    1 && 0 \\
    0 && 1
    \end{pmatrix}.
\end{equation*}

\begin{figure}
    \centering
    \includegraphics[width=8.6cm]{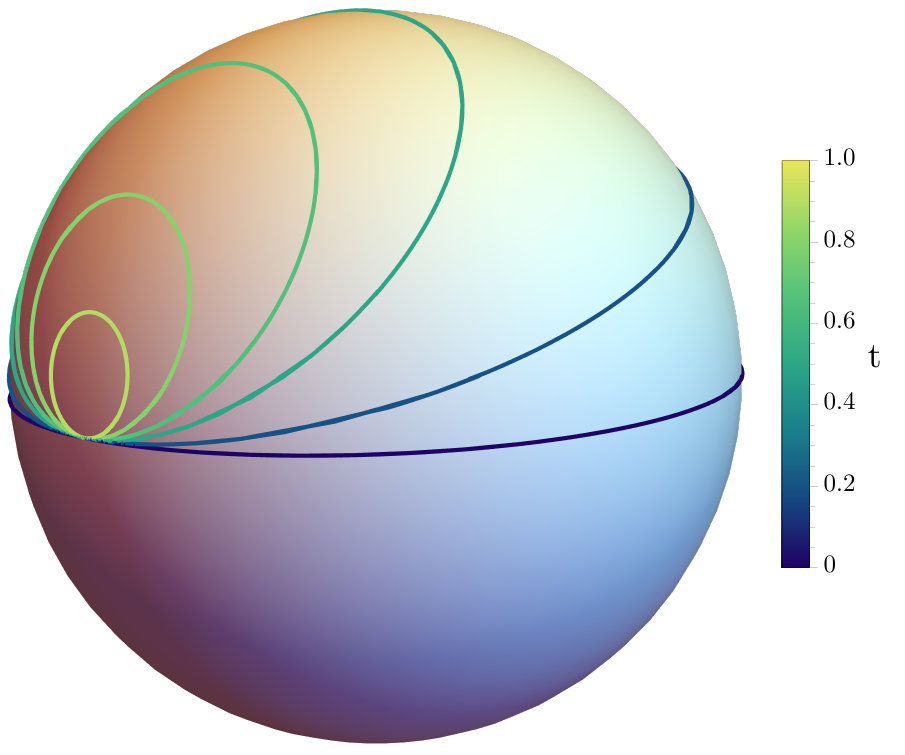}
    \caption{Illustration of the homotopy $F(t,\phi):[0,1]\times S^1 \rightarrow \mathrm{SU}(2)$ from $f(\phi)$ to the identity $\mathds{1}$. The image of $F$ resides on an embedding of $S^2$ in $\mathrm{SU}(2)$. The equator at $t=0$ corresponds to $f(\phi)$ and shrinks down to the identity $(1,0,0)$ at $t=1$.}
    \label{fig:homotopy}
\end{figure}
We construct one such homotopy as follows. $\mathrm{SU}(2)$ is homeomorphic to the unit 3-sphere $S^3$ via
\begin{equation}
    (w,x,y,z) \in S^3 \mapsto \begin{pmatrix}
    x+iz && y+wi \\
    -y+wi && x-iz
    \end{pmatrix} \in \mathrm{SU}(2).
\end{equation}
By setting $w=0$, we obtain a map from $S^2$ into $\mathrm{SU}(2)$:
\begin{equation}
    (x,y,z) \in S^2 \mapsto \begin{pmatrix}
    x+iz && y \\
    -y && x-iz
    \end{pmatrix} \in \mathrm{SU}(2).
\end{equation}
Under this map, $f(\phi)$ corresponds to a loop along the equator in the $xy$-plane traversed twice, while $\mathds{1}$ corresponds to the fixed point $(x,y,z)=(1,0,0)$. By shrinking the equator down to $(1,0,0)$ while remaining on $S^2$, we obtain the desired homotopy as illustrated in Figure \ref{fig:homotopy}. Explicitly, choose
\begin{align}
    x(t,\phi) &= \cos(2\phi) \cos^2\Big(\frac{\pi t}{2}\Big) + \sin^2 \Big(\frac{\pi t}{2} \Big) \label{eq:x}\\
    y(t,\phi) &= \sin(2\phi)\cos\Big(\frac{\pi t}{2}\Big) \label{eq:y}\\
    z(t,\phi) &= \sin^2 (\phi) \, \sin(\pi t). \label{eq:z}
\end{align}
Note that
\begin{equation}
    x(t,\phi)^2 + y(t,\phi)^2 + z(t,\phi)^2 = 1,
\end{equation}
so these points reside on $S^2$. Define the homotopy $F(t,\phi):[0,1]\times S^1\rightarrow \mathrm{SU}(2)$ by
\begin{equation}\label{eq:F_homo}
    F(t,\phi) = \begin{pmatrix}
        x(t,\phi) + iz(t,\phi) && y(t,\phi) \\
        -y(t,\phi) && x(t,\phi) - iz(t,\phi).
    \end{pmatrix}
\end{equation}
Since $F(t,\phi+2\pi) = F(t,\phi)$, $F$ is well-defined as a function on $[0,1]\times S^1$. 

%%%%%%%%%%%% Piecewise smooth frame %%%%%%%%%%%%%%%
\subsubsection{Constructing a piecewise smooth frame for \texorpdfstring{$T^\Comp S^2$}{TCS2}}
On each hemisphere we have a frame, namely, $[v_{1+},v_{2+}]$ and $[v_{1-},v_{2-}]$. We can transform $[v_{1-},v_{2-}]$ along the equator via $f$ so that it agrees with $[v_{1+},v_{2+}]$. If we can extend $f$ to a function $\tilde F(\theta,\phi)$ on all of $D^2_-$, then $\tilde F [v_{1-},v_{2-}]$ would extend $[v_{1+},v_{2+}]$ to a frame on all of $S^2$. The primary difficulty is ensuring that $\tilde F$ is defined at the coordinate singularity $\theta = \pi$. Note that even if the matrix representation of $\tilde F$ is independent of $\phi$ when $\theta= \pi$, it may still be ill-defined at this point since the matrix representation is written in the $(\hat \theta, \hat \phi)$ basis and $\hat \theta$ is not defined at the singularity. One way of avoiding this issue is to require that $\tilde F(\pi,\phi) = \mathds{1}$ since the matrix representation of the identity is basis independent. Thus, we can interpret the desired function $\tilde F(\theta, \phi)$ as a homotopy between $f(\phi)$ and $\mathds{1}$ with $\theta \in [\pi/2,\pi]$ as the homotopy parameter. By reparameterizing from $t \in [0,1]$ to $\theta\in[\pi/2,\pi]$ in Eq.\,(\ref{eq:F_homo}), we obtain the desired function:
\begin{align}\label{eq:rough_homotopy}
    \tilde F(\theta,\phi) &= 
    \begin{pmatrix}
        x(\theta,\phi) + iz(\theta,\phi) && y(\theta,\phi) \\
        -y(\theta,\phi) && x(\theta,\phi) - iz(\theta,\phi)
    \end{pmatrix}, \\
    x(\theta,\phi) &= \cos(2\phi)\sin^2(\theta) + \cos^2(\theta), \\
    y(\theta,\phi) &= \sin(2\phi)\sin(\theta), \\
    z(\theta,\phi) &= -\sin^2(\phi)\sin(2\theta).
\end{align}
Thus, the two vector fields $[v_1,v_2]$ given by $[v_{1+},v_{2+}]$ on $D^2_+$ and $\tilde F(\theta,\phi)[v_{1-},v_{2-}]$ on $D^2_-$ are independent and nonvanishing on all of $S^2$. While $[v_1,v_2]$ are continuous everywhere, they are not smooth at $\theta=\pi/2$, and thus form a continuous, piecewise smooth frame. This is sufficient to prove that $T^\Comp S^2$ is trivial as a topological vector bundle, but not as a smooth vector bundle.

%%%%%%%%%%%%%%% Smoothing the frame %%%%%%%%%%%%%%%%%%%%
\subsubsection{Smoothing the frame}
The last step is to modify the previous construction to form smooth vector fields $[\tilde v_1, \tilde v_2]$. This is done using a smooth step function as follows. Let $g\in C^{\infty}(\mathbb{R})$ be any smooth monotonic function such that $g(t\leq \pi/2)=\pi/2$, $g(t\geq \pi)=\pi$, and $g^{(n)}(\pi/2)=g^{(n)}(\pi)=0$ for all $n\geq 1$. For concreteness, we can define
\begin{equation}
    h(t)= \begin{cases}
        e^{-1/t} & t > 0 \\
        0 & t \leq 0,
    \end{cases}
\end{equation}
and then set
\begin{equation}
    g(\theta) = \frac{\pi}{2}\Big(1 + \frac{h(\theta - \pi/2)}{h(\theta - \pi/2) + h(\pi - \theta)} \Big).
\end{equation}
The desired smooth homotopy is obtained by replacing $\theta$ by $g(\theta)$ in Eq.\,(\ref{eq:rough_homotopy}):
\begin{equation}
    \tilde F_s(\theta,\phi) = \begin{pmatrix}
    x\big( g(\theta),\phi \big) + iz\big( g(\theta),\phi \big) && y\big( g(\theta),\phi \big) \\
    -y\big( g(\theta)\phi \big) && x\big( g(\theta),\phi \big) - iz\big(g(\theta),\phi \big)
    \end{pmatrix}
\end{equation}
Since $g$ is smooth, this replacement preserves smoothness for $\theta \neq \pi/2$. At $\theta = \pi/2$ we must show that
\begin{equation}\label{eq:equitorial_smoothing}
    \partial_{\theta}^{\alpha}\partial_{\phi}^{\beta}[v_{1+},v_{2+}]|_{\theta\rightarrow \pi/2^-} = \partial_{\theta}^{\alpha}\partial_{\phi}^{\beta}\tilde F_s[v_{1-},v_{2-}]|_{\theta\rightarrow \pi/2^+}.
\end{equation}
Since $\hat{\theta}$ and $\hat{\phi}$ are smooth at $\pi/2$, it is sufficient to show Eq. $(\ref{eq:equitorial_smoothing})$ holds when $v_{1,\pm}$ and $v_{2,\pm}$ are expressed as column vectors in the $(\hat{\theta}, \hat{\phi})$ basis. From Eqs.\,(\ref{eq:v1p}-\ref{eq:v2m}), all $\theta$ derivatives of the column vectors $v_{1\pm}$ and $v_{2\pm}$ vanish, so the LHS vanishes for $\alpha \geq 1$. Since also $g^{(\alpha)}(\pi/2)=0$ for $\alpha\geq1$, all $\theta$ derivatives of $\tilde F(\theta, \phi)$ vanish as $\theta\rightarrow \pi/2^+$, so the RHS is also 0 for $\alpha \geq 1$. In the $\alpha = 0$ case, all $\phi$ derivatives agree since
\begin{equation}
    [v_{1+},v_{2+}]_{\theta=\pi/2,\phi} = \tilde F_s[v_{1-},v_{2-}]\big|_{\theta=\pi/2,\phi}.
\end{equation}
This completes the proof that the two vector fields $[v_1,v_2]$ given by $[v_{1+},v_{2+}]$ on $D^2_+$ and $\tilde F_s(\theta,\phi)[v_{1-},v_{2-}]$ form a smooth global frame on $T^\Comp S^2$.

We can use $v_1$ and $v_2$ to write the trivialization of $\gamma$.
\begin{equation}
        \tilde v_{1,2}(\boldsymbol{k}) = v_{1,2} \Big(\frac{\boldsymbol{k}}{|\boldsymbol{k}|}\Big)
\end{equation}
are global nonvanishing sections of $\gamma$, and as such generate trivial line subbundles $\tau_1$ and $\tau_2$ of $\gamma$. By construction
\begin{equation}
    \gamma = \tau_1 \oplus \tau_2.
\end{equation}

\subsection{Vector Potential formulation}
We briefly outline the technical advantages of working with the electric field $\boldsymbol{\mathcal{E}}$ rather than the electromagnetic 4-potential $\mathcal{A}^\mu = (\mathcal{A}^0,\boldsymbol{\mathcal{A}})$ in the vector bundle formulation. Here, $\mathcal{A} = \mathcal{A}(t,\boldsymbol{x})$ and $A = A(k)$ refer to the potential in spacetime and its Fourier transform in 4-momentum space, respectively. We will discuss the two most commonly used gauges, the Lorentz gauge and the Coulomb gauge. In the Lorentz gauge, in which $\partial_\mu \mathcal{A}^\mu = 0$, the potential satisfies the wave equation $\partial_\nu \partial^\nu \mathcal{A}^\mu = 0$ \cite{Tong2006}. Fourier transforming gives $k_\mu A^\mu = 0$ and the dispersion relation $k_\nu k^\nu = 0$, \emph{i.e.}, $k$ is on the lightcone $\Lightcone$. Note that if $A$ were real, the former equation would imply that $(k,A)$ is in the tangent bundle of the lightcone $T\Lightcone$. Indeed, $(k, \epsilon) \in T\Lightcone(k)$ if and only if $(k + \epsilon) \in \Lightcone$ to order $\epsilon$, that is, if
\begin{equation}
    (k_\nu + \epsilon_\nu)(k^\nu + \epsilon ^\nu) = 2 k_\nu\epsilon^\nu = 0
\end{equation}
where we have neglected terms of order $O(\epsilon^2)$. Since $A$ is complex, the modes $(k,A)$ actually form the complexified tangent bundle of the lightcone $T^\Comp \Lightcone$. This is a rank-3 bundle despite there only being two physical wave modes, reflecting the fact that the Lorenz gauge is not complete, and has a remaining degree of redundancy \cite{Tong2006}. This is problematic because two linearly independent modes can represent the same physical wave, making it difficult to draw meaningful conclusions from the topology of this bundle. One can remove the residual gauge freedom by enforcing the Coulomb gauge condition $\boldsymbol{k} \cdot \boldsymbol{A} = 0$, which then implies $A^0 = 0$ \cite{Tong2006}. We then see that in the Coulomb gauge $\boldsymbol{A}$ and $\boldsymbol{E}$ are defined by the same equations, and thus form the same bundle $\gamma$. The analysis in the preceding sections thus also describes the topology of the vector potential bundle in the Coulomb gauge. However, a difficulty arises in discussing the Poincar\'{e} symmetry of Maxwell's equation. The Coulomb gauge is not Lorentz invariant \cite{Tong2006}, and generally a Lorentz transformation will map a mode $(k,A)\in \gamma$ out of the bundle $\gamma$. In fact, applying a Lorentz transformation to all $(k,A) \in \gamma$ will produce a new bundle, which will be isomorphic to $\gamma$. Working directly with this formulation is cumbersome since one must keep track of the different isomorphisms that result from each Lorentz transformation. We thus use the simpler alternative of working with the physically observable field $\boldsymbol{E}$ since Lorentz transformations must keep $(k,\boldsymbol{E}) \in \gamma$ on the bundle $\gamma$, as will be discussed at length in Section \ref{sec:BundleReps}. 

%%%%%%%%%%%%%%%% R- and L- Subbundles %%%%%%%%%%%%%%%%%%%%%%%%%
\section{Topology of circularly polarized subbundles and the nonexistence of linearly polarized subbundles}
\label{sec:TopologyRL}
We have shown that the total photon bundle is trivial, possessing global nonvanishing sections. This decomposes $\gamma$ into two trivial subbundles $\tau_{1,2}$. However, because the base manifold $\mathcal{L}_+$ is not contractible, $\gamma$ can possess topologically nontrivial subbundles. Indeed, we will show that $\gamma$ also decomposes into $R$- and $L$-circularly polarized subbundles which are nontrivial. This decomposition is important because, unlike the trivial decomposition, it is Lorentz invariant, and therefore plays a special role in particle physics.

We discuss three important types of photon polarizations. There are the trivial polarizations constructed in the previous sections. In physics, one frequently encounters linear and circular polarizations as well. We will show that these three types of polarizations are significantly different in terms of their global properties. The trivial polarizations form subbundles $\tau_{1,2}$ with global nonvanishing sections. We will show that the $R$- and $L$-circularly polarized bundles also form well-defined subbundles of $\gamma$, however, they are topologically nontrivial, and therefore do not admit global nonvanishing sections. Linear polarizations, on the other hand, do not even form vector bundles, and are thus not globally well-defined. 

\begin{definition}[Circularly polarized subbundles]\label{def:circ_pol_sub}
    A vector $v \in \gamma(\boldsymbol{k})$ is defined to be $R$-circularly polarized if $v = \alpha (\boldsymbol{e}_1 + i\boldsymbol{e}_2)$ where $\alpha \in \Comp$ and $\boldsymbol{e}_1, \boldsymbol{e}_2$ are real unit vectors in $\gamma(\boldsymbol{k})$ such that $(\boldsymbol{e}_1, \boldsymbol{e}_2, \boldsymbol{k})$ form a right-handed orthogonal coordinate system. The $R$-circularly polarized vectors in $\gamma(\boldsymbol{k})$ form a vector space denoted $\gamma_+(\boldsymbol{k})$. Collected together, these vector spaces form the $R$-circularly polarized vector bundle $\gamma_+ = \amalg_{\boldsymbol{k}} \gamma_+(\boldsymbol{k})$. The $L$-circularly polarized bundle $\gamma_-$ is defined analogously, with $(\boldsymbol{e}_1,\boldsymbol{e}_2,\boldsymbol{k})$ left-handed.
\end{definition}
We show that $\gamma_+$ and $\gamma_-$ are in fact well-defined line subbundles of $\gamma$ via the subbundle criterion (\cite{Tu2017differential}, Theorem 20.4):
\begin{lemma}[Subbundle criterion] \label{lm:subbundle_criterion}
    Let $\pi:E\rightarrow M$ be a smooth rank-$r$ vector bundle and $F=\amalg_{p\in M} F_p$  a subset of $E$ such that for every $p\in M$, the set $F_p$ is a $k$-dimensional vector subspace of the fiber $E_p$. If for every $p\in M$, there exist a neighborhood $U$ of $p$ and $m\geq k$ smooth sections $s_1, \ldots, s_m$ of $E$ over $U$ that span $F_q$ at every point $q \in U$, then $F$ is a smooth subbundle of $E$. 
\end{lemma}
\begin{theorem}
    $\gamma_+$ and $\gamma_-$ are complex line subbundles of $\gamma$ such that $\gamma = \gamma_+ \oplus \gamma_-$. 
\end{theorem}
\begin{proof}
    We first prove that $\gamma_+$ is a subbundle of $\gamma$ using the subbundle criterion with $m=k=1$. For $\boldsymbol{k}_0 \in \pspace$, let $U$ be small ball centered around $\boldsymbol{k}_0$ which does not contain $\boldsymbol{0}$. $U$ is contractible, so there exists a smooth right-handed orthonormal frame $(\boldsymbol{f}_1,\boldsymbol{f}_2, \hat{\boldsymbol{k}})$ of $\gamma|_U$. $s=\boldsymbol{f}_1 + i \boldsymbol{f}_2$ is then a smooth section of $\gamma|_U$ consisting of right-circularly polarized vectors. To show that $s$ spans $\gamma_+(\boldsymbol{k})$ for each $\boldsymbol{k}\in U$, let $\alpha(\boldsymbol{e}_1 + i\boldsymbol{e}_2) \in \gamma(\boldsymbol{k})$ with $\alpha \in \Comp$. Since $(\boldsymbol{e}_1,\boldsymbol{e}_2, \hat{\boldsymbol{k}})$ and $(\boldsymbol{f}_1,\boldsymbol{f}_2, \hat{\boldsymbol{k}})$ are both right-handed, they are related by a rotation by some angle $\theta$ about $\boldsymbol{k}$:
    \begin{equation}
        [\boldsymbol{e}_1, \boldsymbol{e}_2] = [\boldsymbol{f}_1, \boldsymbol{f}_2]\begin{pmatrix}
            \cos \theta & \sin \theta \\
            -\sin \theta & \cos \theta
        \end{pmatrix}
    \end{equation}
    Then at $\boldsymbol{k}$ we have
    \begin{align}
    \begin{split}
        \alpha e^{i\theta}s &= \alpha e^{i\theta}(\boldsymbol{f}_1+i\boldsymbol{f}_2) \\
        &= \alpha[(\cos \theta \, \boldsymbol{f}_1 - \sin \theta\, \boldsymbol{f}_2) + i(\sin \theta\, \boldsymbol{f}_1 + \cos \theta\, \boldsymbol{f}_2)] \\
        &= \alpha(\boldsymbol{e}_1 + i\boldsymbol{e}_2),
    \end{split}
    \end{align}
    showing that $s$ spans $\gamma_+(\boldsymbol{k})$. Thus, $\gamma_+$ is a rank-$1$ subbundle of $\gamma$ by the subbundle criterion, as is $\gamma_-$ by an analogous argument. For fixed $\boldsymbol{k}$, $R$- and $L$-polarizations form a basis for all polarization states. Allowing $\boldsymbol{k}$ to vary then shows that $\gamma = \gamma_+ \oplus \gamma_-$.
\end{proof}

An intuitive way to see that the complex line bundles $\gamma_\pm$ are nontrivial is to consider the underlying rank-2 real vector bundles $\gamma_\pm^\Real$ obtained by forgetting the complex structure on $\gamma_\pm$. That is, $\gamma_\pm^\Real$ consists of the same base space and fibers as $\gamma_\pm$, but scalar multiplication is restricted to $\Real$. 
\begin{theorem}
    The real vector bundles $\gamma_\pm^\Real|_{S^2}$ and $TS^2$ are isomorphic, which implies that $\gamma_+$ and $\gamma_-$ are nontrivial.
\end{theorem}
\begin{proof}
    Let $(\boldsymbol{k},v) \in TS^2$ and let $R_{\boldsymbol{k}}(\frac{\pi}{2})$ denote a $\pi/2$ rotation about $\boldsymbol{k}$ in the positive sense. Define the functions $h_\pm: TS^2 \rightarrow \gamma_\pm^\Real|_{S^2}$ by
    \begin{equation}
        h_\pm(\boldsymbol{k},v) = \Big(\boldsymbol{k}, v \pm iR_{\boldsymbol{k}}\Big(\frac{\pi}{2}\Big)v\Big).
    \end{equation}
    $h_\pm$ is a clearly a smooth, injective bundle map of real vector bundles over $S^2$. To show that it is also surjective, and thus an isomorphism, note that by Definition \ref{def:circ_pol_sub} any vector  $u_\pm \in \gamma_\pm^\Real|_{S^2}(\boldsymbol{k})$ can be expressed as 
    \begin{equation}
        u_\pm = \alpha\Big(w \pm iR_{\boldsymbol{k}}\Big(\frac{\pi}{2}\Big)w \Big)
    \end{equation}
    for some $w \in TS^2(\boldsymbol{k})$ and $\alpha = (a+ib) \in \mathbb{C}$. This can be rewritten as
    \begin{align}
        u_\pm &= w' \pm iR_{\boldsymbol{k}}w' = h_\pm(w')\\
        w' &\doteq aw \mp b R_{\boldsymbol{k}}\Big(\frac{\pi}{2}\Big) w \in TS^2(\boldsymbol{k}),
    \end{align}
    showing that $h_\pm$ is surjective, and thus $\gamma_\pm^\Real|_{S^2} \cong TS^2$.

    If $\gamma_\pm$ were trivial there would exist a nonvanishing section of $s$ of $\gamma_\pm|_{S^2}$, which can then be considered as a nonvanishing section of $\gamma_\pm^\Real|_{S^2} \cong TS^2$ which contradicts the hairy ball theorem.
\end{proof}

While the previous theorem shows that $\gamma_\pm$ are nontrivial, it remains to classify them as complex line bundles. As with $\gamma$, it is easier to work with the restrictions $\gamma_{\pm} |_{S^2}$. $\gamma_+|_{S^2}$ and $\gamma_-|_{S^2}$ are quite similar as they are conjugate bundles $\gamma_+|_{S^2} = \bar{\gamma}_-|_{S^2}$. Furthermore, they are isomorphic when considered as real vector bundles by the previous theorem. However, they are not actually isomorphic as complex vector bundles, as the complex structure detects their different orientations. To see this, note that by Lemma \ref{lm:complexification_lemma}, the first Chern number of a line bundle and its conjugate differ by a sign:
\begin{equation}\label{eq:conjugate_chern}
    C_1(\gamma_+|_{S^2}) = -C_1(\gamma_-|_{S^2}).
\end{equation}
Complex line bundles over $S^2$ are completely classified by their first Chern number $C_1$ \cite{McDuff2017}, so $C_1(\gamma_\pm|_{S^2}) \neq 0$ as these are nontrivial bundles. Eqn. (\ref{eq:conjugate_chern}) then implies these bundles are not isomorphic. The next result fully classifies $\gamma_\pm|_{S^2}$. By the correspondence of bundles over $S^2$ and $\pspace$, this also classifies $\gamma_\pm$.

\begin{theorem}
     $C_1(\gamma_\pm|_{S_2}) = \mp 2$.
\end{theorem}
\begin{proof}
We calculate the Chern numbers using the Berry curvature. The Hermitian structure along with the identification of the electric fields of the fibers with vectors in $\mathbb{C}^3$ gives a natural Berry connection. For any local frame $e_\pm(\boldsymbol{k})$ of $\gamma_\pm|_{S^2}$ in some neighborhood $U$, the Berry connection form is \cite{Frankel2011}
\begin{equation}
    \omega_\pm = \langle e_\pm, de_\pm \rangle.
\end{equation}
On the neighborhood $U$ consisting of all unit $\hat{\boldsymbol{k}}$ not on the $z$-axis, we can choose the frames for $\gamma_\pm|_{S^2}$ to be
\begin{equation}\label{eq:polarization_basis}
    e_\pm = \frac{1}{\sqrt{2}}(\hat \theta \pm i \hat \phi).
\end{equation}
Then,
\begin{equation}\label{eq:berry_connection_RL}
    \omega_\pm = \mp i \cos \theta \, d\phi.
\end{equation}
and the curvature form is given on $U$ by
\begin{equation}\label{eq:berry_curvature_RL}
    \Omega_\pm = d\omega_\pm = \pm i \sin \theta \, d\theta \wedge d\phi = \pm i\, dA
\end{equation}
where $dA$ is the Euclidean area form on $S^2$. Unlike the connection forms, the curvature forms are always globally defined. Since $\Omega_\pm = \pm i dA$ on a dense neighborhood $U$, this relation holds on all of $S^2$. The first Chern number is then given by \cite{Tu2017differential}
\begin{equation}
    C_1(\gamma_\pm|_{S^2}) = \frac{i}{2\pi}\int_{S^2} \Omega_\pm = \mp 2.
\end{equation}
Since $C_1$ is nonzero, $\gamma_\pm$ are nontrivial.
\end{proof}
One consequence of the nontriviality of $\gamma_\pm$ can be found in the momentum space representation of the photon wavefunction \cite{Bialynicki1996,Bialynicki-Birula2011}. In this representation, the wavefunction of a photon is written as a two component wavefunction, with an $R$- and $L$-part
\begin{equation}
    \boldsymbol{f}(\boldsymbol{k}) = \begin{pmatrix}
        \boldsymbol{f}_+(\boldsymbol{k}) \\
        \boldsymbol{f}_-(\boldsymbol{k})
    \end{pmatrix}
     = \begin{pmatrix}
         f_+(\boldsymbol{k})e_+(\boldsymbol{k}) \\
         f_-(\boldsymbol{k})e_-(\boldsymbol{k})
     \end{pmatrix}
\end{equation}
where $e_+(\boldsymbol{k})$ is a choice of unit $R$-polarization vector and $e_-(\boldsymbol{k}) \doteq e_+^*(\boldsymbol{k})$ is a corresponding $L$-polarization. Frequently, the vectors $e_\pm$ are left implicit, and the wavefunction is written in terms of the two scalar functions $f_\pm$ as
\begin{equation}
    \boldsymbol{f} = \begin{pmatrix}
        f_+(\boldsymbol{k}) \\
        f_-(\boldsymbol{k})
    \end{pmatrix}.
\end{equation}
However, the $e_\pm$ are sections of the nontrivial bundles $\gamma_\pm$, and thus no continuous $e_\pm$ exist! If $S_a$ denotes the sphere of radius $a>0$ in momentum space $\pspace$, then $e_+$ must be discontinuous for at least one point $\boldsymbol{k}_0(a)$ on each $S_a$. For $\boldsymbol{f}(\boldsymbol{k})$ to be continuous at $\boldsymbol{k}_0(a)$, the component functions must satisfy $f_\pm(\boldsymbol{k}_0(a)) = 0$. Thus, the photon wavefunction components $f_\pm$ are not free scalar functions. They must obey the topological constraint of having a zero on every closed surface enclosing $\boldsymbol{k}=0$. 

Our calculation of $C_1(\gamma_\pm)$ is very similar to one done by Bliokh \emph{et al.}  \cite{Bliokh2015}, although they use a different sign convention for the Berry connection and Chern number, resulting in an overall sign difference. Our approach differs in that it emphasizes the underlying vector bundle structures, and in particular, shows that $\gamma_\pm$ are in fact well-defined bundles. This is important, for example, because by analogy one might assume that linearly polarized photons also possess some topological structure. However, we show that unlike the circularly polarized subbundles $\gamma_\pm$ or the trivial subbundles $\tau_{1,2}$, there are no linearly polarized subbundles. 
\begin{definition}[Linearly polarized subbundle]
    Linearly polarized vectors in $\gamma$ are those of the form $\alpha e$ where $\alpha \in \Comp$ and $e \in \mathbb{R}^3$; a rank 1 subbundle is linearly polarized if it only contains linearly polarized vectors.
\end{definition}
\begin{theorem}\label{thm:no_linear_subbundles}
    There are no linearly polarized subbundles of $\gamma$. 
\end{theorem}

To prove Theorem \ref{thm:no_linear_subbundles} we use the following known result.
\begin{lemma}\label{lm:no_line_bundles}
    Any real line bundle over a simply connected space is trivial. In particular, all real line bundles over $S^2$ are trivial.
\end{lemma}
\begin{proof}
    A real line bundle is trivial if and only if it is orientable \cite{Encyc2006}. The result then follows from the fact that every real line bundle over a simply connected space is orientable (\cite{Bott2013}, Prop. 11.5).
\end{proof}

We can now prove Theorem \ref{thm:no_linear_subbundles} via a contradiction:
\begin{proof}[Proof of Theorem \ref{thm:no_linear_subbundles}]
    Suppose $\gamma_{lin}$ is a linearly polarized line subbundle of $\gamma$. Then $\tilde \gamma_{lin} \doteq  \gamma_{lin}|_{S^2}$ is a subbundle of $\gamma|_{S^2} \cong T^\Comp S^2$. Each fiber $\tilde \gamma_{lin}(\hat{\boldsymbol{k}})$ contains a $1$D real subspace $\gamma_\mathbb{R}(\boldsymbol{k})$. The collection of these subspaces $\gamma_\mathbb{R} = \amalg_{\hat{\boldsymbol{k}}}\gamma_\mathbb{R}(\hat{\boldsymbol{k}})$ forms a real line bundle over $S^2$ which is a subbundle of $TS^2$. By Corollary \ref{lm:no_line_bundles}, $\gamma_\mathbb{R}$ is trivial and must therefore have a nonvanishing section $s$. However, $s$ would also then be a nonvanishing section of $TS^2$, violating the hairy ball theorem.
\end{proof}

Although the the total photon bundle $\gamma$ is topologically trivial, the $R$- and $L$-circularly polarized photons form nontrivial subbundles. This nontriviality embedded within $\gamma$ can be detected 
in the associated frame bundle $\mathcal{F}(\gamma)$ of $\gamma$ by reducing the structure group. $\mathcal{F}(\gamma)$ is a principal bundle over $\pspace$ where the local sections consist of local frames of $\gamma$. The most general structure group of $\mathcal{F}(\gamma)$ is $\mathrm{GL}(2,\Comp)$. By considering only orthonormal frames, the structure group can be reduced to $\mathrm{U}(2)$. We showed in the previous section that $[v_1,v_2]$ forms a global orthonormal frame of $\gamma$, so the $\mathrm{U}(2)$-frame bundle is trivial. However, we can further reduce the structure group by considering local frames of the form $[w_+,w_-]$ where $w_\pm$ are local orthonormal sections of $\gamma_\pm$. The linear transformations preserving such frames have the form
\begin{equation}
    \begin{pmatrix}
        e^{i\theta_1} & 0 \\
        0 & e^{i\theta_2}
    \end{pmatrix}
\end{equation}
for real $\theta_1, \theta_2$, that is, the structure group is reduced to $\mathrm{U}(1) \times \mathrm{U}(1)$. There are no global sections of this form, so the $\mathrm{U}(1) \times \mathrm{U}(1)$-frame bundle is nontrivial. 

We have shown that $\gamma$ decomposes into both the trivial bundles $\tau_{1,2}$ and the nontrivial subbundles $\gamma_\pm$. These are not, however, the only possible decompositions. In fact, $\gamma$ has an infinite number of nontrivial decompositions by the following observation. The complex line bundles over $S^2$ are labeled by the first Chern number \cite{McDuff2017}. By the correspondence of bundles over $S^2$ and $\pspace$, the same is true of line bundles over $\pspace$. Denote the complex line bundle over $\pspace$ with Chern number $j \in \mathbb{Z}$ by $\ell_j$.
\begin{theorem}\label{thm:infinte_subbundles}
    $\gamma \cong \ell_j \oplus \ell_{-j}$ for every $j \in \mathbb{Z}$. In particular, $\gamma$ has a subbundle isomorphic to $\ell_j$ for every $j$.
\end{theorem}
\begin{proof}
    By the Whitney product formula \cite{Tu2017differential}, the first Chern number is additive, meaning that if $E_1$ and $E_2$ are vector bundles, then
    \begin{equation}
        C_1(E_1 \oplus E_2) = C_1(E_1) + C_1(E_2).
    \end{equation}
    Thus,
    \begin{equation}
        C_1(\ell_j \oplus \ell_{-j}) = 0.
    \end{equation}
    As we previously argued, rank-2 vector bundles over $\pspace$ are also labeled by their first Chern number. Thus, $\ell_j \oplus \ell_{-j}$ is trivial and therefore isomorphic to $\gamma$.
\end{proof}
Despite the existence of these other nontrivial subbundles, only the circularly polarized subbundles $\gamma_\pm \cong \ell_{\mp2}$ are typically used in applications. From a practical standpoint this is easy to understand, as $\gamma_\pm$ have simple descriptions, while finding explicit expressions for subbundles of $\gamma$ isomorphic to $\ell_j$ is a nontrivial task. However, from a purely topological perspective, there is no reason to prefer the circularly polarized subbundles over the other $\ell_j$. We will show in the next section that $\gamma_\pm$ are special from a geometric perspective when we consider the role of Poincar\'{e} symmetry.

%%%%%%%%% Section 4: Classification of photons by vector bundle representations of the Poincar\'{e} group %%%%%%%%%%%%%
\section{Classification of photons and other massless particles by vector bundle representations of the Poincar\'{e} group}
\label{sec:BundleReps}

In this section we show how Poincar\'{e} symmetries of vector bundles can be used to classify massless particles such as photons. In quantum field theory, particles are unitary irreducible representations of the Poincar\'{e} group on a Hilbert space. We restrict our discussion to non-projective representations since this includes the photon case. The conventional method for constructing and classifying such representations is via Wigner's little group method \cite{Weinberg1995}. While this method produces smooth representations of massive particles, obstructions form in the massless case with the space of single-particle states becoming discontinuous \cite{Flato1983,Dragon2022}, resulting in non-smooth representations on the Hilbert space of wavefunctions. By this we mean Poincar\'{e} transformations generally map smooth wavefunctions to non-smooth wavefunctions. We will see that this issue traces back to a topological singularity that occurs in the $m\rightarrow 0$ limit as the mass hyperboloid of momentum space becomes the topologically nontrivial lightcone. We will show that by considering vector bundle representations on the single-particle states rather than vector space representations, one can obtain globally well-defined representations on the single-particle states. We prove that massless unitary irreducible vector bundle representations of the Poincar\'{e} group naturally induce unitary irreducible representations on the Hilbert space of $L^2$ bundle sections, and therefore such bundle representations correspond to particles by the usual definition. We show that Wigner's little group method generalizes to this bundle formalism, and can be used to decompose bundle representations into irreducible representations and classify these irreducible representations. In the case of the photon bundle, this induces the decomposition $\gamma = \gamma_+ \oplus \gamma_-$, showing that $R$- and $L$-photons are globally well-defined particles. An important implication of this research is that a given global photon wave function can be uniquely decomposed into $R$- and $L$-photon components, even though there exist no global bases for the $R$- and $L$-photons.  

Our results complement the body of work which has further developed Wigner's original version of the little group method. Mackey \cite{Mackey1952,Mackey1953,Mackey1968} developed the mathematician's version of this theory which generalized the little group method from the Poincar\'{e} group to any locally compact group. Simms \cite{Simms1968} showed that Mackey's theory had an underlying vector bundle structure, using what he called $G$-Hilbert bundles which are very similar to the modern notion of equivariant vector bundles which we use. However, Simms worked with only topological vector bundles, not smooth vector bundles, and as such his theory did not resolve the non-smoothness problem with Wigner's representations. Indeed, Simms's bundle formalism did not offer obvious advantages over Mackey's theory, and Mackey referred to the bundle formulation as a ``digression'' \cite{Mackey1968}. However, we show that our bundle formalism produces smooth representations for massless particles. Flato \emph{et al.} \cite{Flato1983} also noted singularities in the standard massless representations of the Poincar\'{e} group. They remedied this issue by developing a theory of twisted delta (``twelta'') functions, which they use to represent single-particle states. This produced smooth vector space representations on the functionals of sections of vector bundles. They primarily use field-theoretic techniques, with line bundles appearing implicitly via transformation properties of functions. Dragon \cite{Dragon2022} also recently approached this problem, finding massless field representations with underlying vector bundle structures. The formalism we develop is based on the notion of equivariant vector bundles and differs from previous methods in that it considers Poincar\'{e} actions directly on smooth vector bundles. This is useful, as we have seen in the case of Maxwell's equations that single-particle states (eigenmodes) naturally form a vector bundle $\gamma$.

This section is organized as follows. In part \ref{subsec:geomtric_notation} we define our notation and conventions. In part \ref{subsec:Representations_on_vector_bundles}, we discuss vector bundle representations of symmetry groups. In mathematics, these bundle representations are known as equivariant vector bundles. Since equivariant vector bundles are rarely used in physics, we will develop the necessary theory here to discuss massless particles.  In part \ref{subsec:Little_group_singularity}, we discuss the issues with the conventional little group method for constructing irreducible representations for massless particles. In part \ref{subsec:bundle_little_group}, we show that these issues can be overcome by developing a modified version of the little group method for bundle representations of massless particles. In particular, we will see that $R$- and $L$-photons are globally well-defined as irreducible bundle representations of the Poincar\'{e} group, and are thus globally well-defined. In part \ref{subsec:sectional_reps} we prove that unitary irreducible massless bundle representations induce unitary irreducible Hilbert space representations, and thus correspond to particles in the usual sense. 

%%%%%%%%% Geoemtric notation %%%%%%%%%%%%
\subsection{Geometric notation and conventions}\label{subsec:geomtric_notation}
We use the $(-,+,+,+)$ signature for Minkowski space, with the minus sign on the time component. The spacetime metric is denoted by 
\begin{equation}
    \eta^{\mu \nu} = \eta_{\mu \nu} = 
    \begin{pmatrix}
        -1 & 0 & 0 & 0 \\
         0 & 1 & 0 & 0 \\
         0 & 0 & 1 & 0 \\
         0 & 0 & 0 & 1
    \end{pmatrix}.
\end{equation} 
The Poincar\'{e} group, consisting of all isometries of Minkowski space, is denoted by $\mathrm{IO}(3,1)$, where the ``I'' indicates the inclusion of inhomogeneous transformations. Likewise, the Euclidean group $\mathrm{ISO}(n) \cong \mathbb{R}^n \rtimes \mathrm{SO}(n)$ is the group of all isometries of $n$-dimensional Euclidean space and $\rtimes$ denotes the semidirect product. The Lorentz group $\mathrm{O}(3,1)$ is the subset of $\mathrm{IO}(3,1)$ consisting of only homogeneous isometries. 
A generic element of $\mathrm{IO}(3,1)$ can be expressed uniquely as $\mathsf{a} \circ \Lambda$, where $\mathsf{a}$ is a spacetime translation and $\Lambda \in \mathrm{O}(3,1)$. The orthochronous Poincar\'{e} and Lorentz groups, $\mathrm{IO}^+(3,1)$ and $\mathrm{O}^+(3,1)$, consist of transformations preserving the orientation of time:
\begin{align}
    \mathrm{IO}^+(3,1) = \{\mathsf{a} \circ \Lambda \in \mathrm{IO}(3,1)| {\Lambda^0}_0 \geq +1 \} \\
    \mathrm{O}^+(3,1) = \{\Lambda \in \mathrm{O}(3,1)| {\Lambda^0}_0 \geq +1 \}.
\end{align}
The proper orthochronous Poincar\'{e} and Lorentz groups, $\mathrm{ISO}^+(3,1)$ and $\mathrm{SO}^+(3,1)$, are the further restrictions of $\mathrm{IO}^+(3,1)$ and $\mathrm{O}^+(3,1)$ to transformations preserving the orientation of space, that is, with $\mathrm{det} \, \Lambda = 1$. The parity and time inversion operators are the Lorentz transformations defined by their action on the 4-vector $x = (x^0, \boldsymbol{x})$ by
\begin{align}
\mathsf{P}(x^0,x^i) &= (x^0,-x^i) \\
\mathsf{T}(x^0,x^i) &= (-x^0, x^i). 
\end{align}
$\mathrm{SO}^+(3,1)$ is the connected identity component of $\mathrm{O}(3,1)$ and contains neither $\mathsf{P}$ nor $\mathsf{T}$. $\mathrm{O}^+(3,1)$ contains $\mathsf{P}$ but not $\mathsf{T}$.

%%%%%%%%% Representations on vector bundles %%%%%%%%%%%%%%
\subsection{Representations on vector bundles}\label{subsec:Representations_on_vector_bundles}
Representation theory is the method of representing groups by their actions on a linear space. Typically this linear space is taken to be a vector space. However, vector bundles also possess a linear structure, albeit  one that is more complicated than that of a vector space. We now define a representation of a Lie group on a vector bundle:
\begin{definition}[Equivariant Vector Bundle]
Let $G$ be a Lie group. A $G$-equivariant vector bundle is a vector bundle $\pi :E\rightarrow M$ equipped with a smooth group action $\Sigma: G\times E \rightarrow E$ such that for each $g\in G$, $\Sigma_g=\Sigma(g,\cdot):E\rightarrow E$ is a vector bundle isomorphism. In particular, there is an induced diffeomorphism $\tilde \Sigma_g:M\rightarrow M$ such that
\begin{equation} \begin{tikzcd}
    E \arrow[r,"\Sigma_g"] \arrow[d, swap, "\pi"] & E \arrow[d,"\pi"] \\
M \arrow[r,"\tilde{\Sigma}_g"] & M
\end{tikzcd} \end{equation}
commutes. We also say that $(E,\Sigma)$ is a representation of $G$ on the vector bundle E. The action of $g$ on a vector $v \in E$ or basepoint $m$ is sometimes written with the shorthand $gv \doteq \Sigma_g v$  or $gm \doteq \tilde \Sigma_g m$ when there is no chance of ambiguity. We also frequently use the notation $\Sigma (g) \doteq \Sigma_g$.

$E$ is said to be a homogeneous bundle representation if the action $\tilde \Sigma$ on $M$ is transitive. If E is a Hermitian bundle, then we say the representation is unitary if the restriction of $\Sigma_g$ to any fiber is unitary, \emph{i.e.}, if  $\Sigma_g|_{E_p}:E_p\rightarrow E_{\tilde{\Sigma}_g(p)}$ is unitary. If $F$ is a subbundle of $E$ with nonzero rank, then we say $(F,\Sigma)$ is a subrepresentation of $(E, \Sigma)$ if $\Sigma$ restricts to an action of $G$ on $F$. $E$ is irreducible if it has no proper subrepresentations.
\end{definition}

\begin{definition}[Isomorphism of equivariant vector bundles]\label{def:equivariant_iso}
    An isomorphism between two $G$-equivariant vector bundles $\pi_1:E_1 \rightarrow M$ and $\pi_2:E_2 \rightarrow M$ with group actions $\Sigma_1$ and $\Sigma_2$ is a vector bundle isomorphism $h: E_1 \rightarrow E_2$ which preserves base points,
    \begin{equation}
        h(E_1(m)) = E_2(m)
    \end{equation}
    for $m \in M$, and which satisfies the equivariance property
    \begin{equation}
        h\big(\Sigma_1(g)(m,v)\big) = \Sigma_2(g)h(m,v)
    \end{equation}
    for any $g \in G$ and $(m,v) \in E_1(m)$. $h$ is said to be a unitary isomorphism if the actions $\Sigma_1$ and $\Sigma_2$ are unitary, and if $h$ is unitary.
    If such an isomorphism exists, the representations are said to be equivalent or unitarily equivalent.
\end{definition}
These definitions are analogous to those for ordinary representations, with vector spaces and linear maps replaced by vector bundles and bundle maps. As for ordinary representations, a representation of a Lie group $G$ on a vector bundle $E$ describes a $G$-symmetry of the system represented by $E$. We show that the photon bundle $\gamma$ is unitary and equivariant with respect to the orthochronous Poincar\'{e} group, $\text{IO}^+(3,1)$, describing the Poincar\'{e} symmetry of the Maxwell system.  We note that it is possible to construct a bundle representation of the full Poincar\'{e} group $\mathrm{IO}(3,1)$ if one includes the backward light cone in the base manifold of $\gamma$. However, little is gained by doing so and it introduces two complications. First, the full light cone is not connected while the forward light $\mathcal{L}_+$ cone is. Second, the representation would no longer be unitary since the action of time-inversion is anti-unitary \cite{Bargmann1954}.

Let $(k,\boldsymbol{E}) \in \gamma$ where $k= (|\boldsymbol{k}|,\boldsymbol{k})$ is a lightlike 4-vector. Both $k$ and $\boldsymbol{E}$ have well-known transformation laws under Lorentz transformations. For $\Lambda \in \mathrm{O}(3,1)$, $k$ transforms like a 4-vector: $k\rightarrow k'= (|\boldsymbol{k}'|,\boldsymbol{k}') = \Lambda k$. As $\Lightcone$ can be parameterized by $\boldsymbol{k}$, we sometimes write $\boldsymbol{k}' = \Lambda \boldsymbol{k}$ for the spatial part of $k'$. For equivariant bundles with base manifold a subset of 4-momentum space, we will always assume that the action $\tilde \Sigma$ on the base manifold is the 4-vector action. Under the subset of $\mathrm{O}(3,1)$ consisting of rotations in $3$-space, $\boldsymbol{E}$ transforms like a 3-vector, since 3D rotations and the Fourier transform commute. Let $\Lambda_{\boldsymbol{v}}$ denote a boost by velocity $\boldsymbol{v}$. Under this transformation, the electric field $\xSpace{E}(x)$ becomes \cite{Jackson1999}
\begin{align}
\begin{split}
    \xSpace{E}'(\Lambda_{\boldsymbol{v}}x) &= \gamma_{Lor}[\xSpace{E}(x) + \boldsymbol{v} \times \xSpace{B}(x)] - \frac{\gamma_{Lor}^2}{\gamma_{Lor} + 1} (\xSpace{E}(x)\cdot \boldsymbol{v})\boldsymbol{v}, \\
    \gamma_{Lor} &= \frac{1}{\sqrt{1-v^2}}.
\end{split}
\end{align}
Then,
\begin{align}
\begin{split}
    \kOmegaSpace{E}'(\Lambda_{\boldsymbol{v}} k) &= \int e^{i(\Lambda_{\boldsymbol{v}} k)^\mu x'_\mu}\xSpace{E}'(x')\, d^4x' \\
    &= \int e^{i (\Lambda_{\boldsymbol{v}}k)^\mu (\Lambda_{\boldsymbol{v}}x)_\mu} \xSpace{E}'(\Lambda_{\boldsymbol{v}} x) \det (\Lambda_{\boldsymbol{v}}) \,d^4 x \\
    &= \int e^{ik^\mu x_\mu}\xSpace{E}'(\Lambda_{\boldsymbol{v}} x)\,d^4 x \\
    &= \gamma_{Lor}[\kOmegaSpace{E}(k) + \boldsymbol{v} \times (\boldsymbol{\hat{k}} \times \kOmegaSpace{E}(k))] - \frac{\gamma_{Lor}^2}{\gamma_{Lor} + 1} (\kOmegaSpace{E}(k)\cdot \boldsymbol{v})\boldsymbol{v}. \label{eq:E_tilde_transform}
\end{split}
\end{align}
Note that for two $4$-vectors $a^\mu = (a^0,\boldsymbol{a})$ and $b^\mu = (b^0,\boldsymbol{b})$, both $\delta^4(a^\mu -b^\mu)=\delta(a^0-b^0)\delta^{3}(\boldsymbol{a}-\boldsymbol{b})$ and $a^0\delta^3(\boldsymbol{a}-\boldsymbol{b})$ are Lorentz invariant \cite{Weinberg1995}, and therefore so is $(a^0)^{-1}\delta(a^0-b^0)$:
\begin{equation}\label{eq:delta_transform}
    \frac{\delta\Big((\Lambda a)^0 - (\Lambda b)^0\Big)}{(\Lambda a)^0} = \frac{\delta(a^0 - b^0)}{a^0}
\end{equation}
for any Poincar\'{e} transformation $\Lambda$. It then follows from  Eqs. (\ref{eq:E_Fourier}), (\ref{eq:E_tilde_transform}), and (\ref{eq:delta_transform}) that
\begin{equation}
    \boldsymbol{E}(\Lambda_{\boldsymbol{v}}\boldsymbol{k}) = \frac{|\boldsymbol{k}|}{|\Lambda_{\boldsymbol{v}} \boldsymbol{k}|}\Big[\gamma_{Lor}[\boldsymbol{E}(\boldsymbol{k}) + \boldsymbol{v} \times (\boldsymbol{\hat{k}} \times \boldsymbol{E}(\boldsymbol{k}))] - \frac{\gamma_{Lor}^2}{\gamma_{Lor} + 1} (\boldsymbol{E}(\boldsymbol{k})\cdot \boldsymbol{v})\boldsymbol{v} \Big],
\end{equation}
and thus the action of $\Lambda_{\boldsymbol{v}}$ on $(\boldsymbol{k},\boldsymbol{E}) \in \gamma$ is given by
\begin{gather}
    \Sigma_{\Lambda_{\boldsymbol{v}}}(\boldsymbol{k},\boldsymbol{E}) = (\Lambda_{\boldsymbol{v}}k,\boldsymbol{E}') \label{eq:boost}\\
    \boldsymbol{E}' \doteq \frac{|\boldsymbol{k}|}{|\Lambda_{\boldsymbol{v}} \boldsymbol{k}|}\Big[\gamma_{Lor}[\boldsymbol{E} + \boldsymbol{v} \times (\boldsymbol{\hat{k}} \times \boldsymbol{E})] - \frac{\gamma_{Lor}^2}{\gamma_{Lor} + 1} (\boldsymbol{E}\cdot \boldsymbol{v})\boldsymbol{v} \Big].\label{eq:E_prime}
\end{gather}
Similar calculations show that spatial inversion $\mathsf{P}$ and spacetime translations $\mathsf{a}=(\mathsf{a}^0,\boldsymbol{\mathsf{a}})$ act via
\begin{gather}
    \Sigma_\mathsf{P}(\boldsymbol{k},\boldsymbol{E})=(-\boldsymbol{k},-\boldsymbol{E}), \\
    \Sigma_\mathsf{a}(\boldsymbol{k},\boldsymbol{E}) = (\boldsymbol{k},e^{i(\boldsymbol{k}\cdot \boldsymbol{\mathsf{a}}- |\boldsymbol{k}|\mathsf{a}^0)}\boldsymbol{E}).
\end{gather}

Any element of $\mathrm{IO}^+(3,1)$ can be expressed as a combination of boosts, spatial rotations, spacetime translations, and spatial inversion, so together these relations define an action of $\mathrm{IO}^+(3,1)$ on $\gamma$. An essential property is that $k$ transforms independently of $\boldsymbol{E}$, and thus the Poincar\'{e} action preserves fibers: $\Sigma(\Lambda)\gamma(k) \subseteq \gamma({\tilde{\Sigma}(\Lambda)k}) = \gamma(\Lambda k)$. This ensures that elements of $\mathrm{IO}^+(3,1)$ are represented by bundle maps, and therefore this action gives an equivariant bundle structure on $\gamma$. We have thus proved most of the following theorem.
\begin{theorem}\label{thm:gamma_equivariant}
    $\gamma$ is a homogeneous $\mathrm{IO}^+(3,1)$-equivariant vector bundle under the usual Poincar\'{e} transformations of the electric field described above. This representation is unitary with respect to the Hermitian product (\ref{eq:Hermitian_product_E}). 
\end{theorem}
\begin{proof}
    The representation is homogeneous since $\mathrm{IO}^+(3,1)$ acts transitively on the base manifold $\mathcal{L}_+$. It remains to prove unitarity. The Hermitian product induces a norm on each fiber $\gamma(\boldsymbol{k})$ given by $\boldsymbol{E}| = \sqrt{\boldsymbol{E}^*\cdot \boldsymbol{E}}$. The inner product of two arbitrary vectors can be expressed purely in terms of the norm via the polarization identity \cite{Schechter1996}:
    \begin{equation}
        \langle \boldsymbol{E}_1, \boldsymbol{E}_2 \rangle = \frac{1}{4}\big(|\boldsymbol{E}_1 + \boldsymbol{E}_2|^2 - |\boldsymbol{E}_1 - \boldsymbol{E}_2|^2 - i|\boldsymbol{E}_1 + i \boldsymbol{E}_2|^2 + i|\boldsymbol{E}_1 - i\boldsymbol{E}_2|^2 \big)
    \end{equation}
    It is thus sufficient to check that
    \begin{equation}\label{eq:isometry}
        \langle \Sigma_\Lambda \boldsymbol{E}, \Sigma_\Lambda \boldsymbol{E} \rangle = \langle \boldsymbol{E}, \boldsymbol{E} \rangle
    \end{equation}
    for all $\Lambda \in \mathrm{IO}^+(3,1)$. Furthermore, it is sufficient to check the cases when $\Lambda$ is a spatial rotation, spatial inversion, spacetime translation, or boost. The first three are trivial, and for the latter one can show directly from Eq. (\ref{eq:E_prime}) that
    \begin{equation}
        \boldsymbol{E}'^* \cdot \boldsymbol{E}' = \boldsymbol{E}^* \cdot \boldsymbol{E}
    \end{equation}
    which completes the proof.
\end{proof}

As in the case of vector space representations, a representation $\Sigma$ of a Lie group $G$ on a vector bundle $\pi: E\rightarrow M$ induces an action $\sigma$ of the Lie algebra $\text{Lie}(G)$ on that bundle in the following sense. For $\mathfrak{g}\in\text{Lie}(G)$, the generator $\sigma_\mathfrak{g}$ associates to each vector $(m,\boldsymbol{v}) \in E$ a tangent vector in $T_{(m,\boldsymbol{v})}E$ describing the infinitesimal group action. Explicitly,
\begin{equation}
\sigma_\mathfrak{g}(m,\boldsymbol{v}) \doteq \frac{d}{dt}\Big|_{t=0}\Sigma_{\text{Exp}(t\mathfrak{g})}(m,\boldsymbol{v}) \in T_{(m,\boldsymbol{v})}E.
\end{equation}
Thus, $\sigma_\mathfrak{g}\in \mathfrak{X}(E)$ where $\mathfrak{X}(E)$ is the set of vector fields on $E$. The next result says that $\sigma_\mathfrak{g}$ respects the vector bundle structure of $E$ in the sense that it is a lift of the corresponding vector field on the base manifold $M$ induced by $\tilde \Sigma$. Define the vector field $\tilde \sigma \in \mathfrak{X}(M)$ by 
\begin{equation}
    \tilde \sigma_\mathfrak{g}(m) \doteq \frac{d}{dt}\Big|_{t=0}\tilde \Sigma_{\text{Exp}(t\mathfrak{g})}(m) \in T_{m}M.
\end{equation}
\begin{proposition}
\begin{equation}
    \tilde \sigma_\mathfrak{g} \circ \pi  = \pi_* \circ \sigma_\mathfrak{g}
\end{equation}
where $\pi_*$ is the pushforward of $\pi$.  That is, the diagram
\begin{equation} \begin{tikzcd}
    E \arrow[r,"\sigma_\mathfrak{g}"] \arrow[d, swap, "\pi"] & TE \arrow[d,"\pi_*"] \\
M \arrow[r,"\tilde{\sigma}_\mathfrak{g}"] & TM
\end{tikzcd} \end{equation}    
commutes. 
\end{proposition}

\begin{proof}
    $\Sigma_{\text{Exp}(t\mathfrak{g})}(m,\boldsymbol{v})$ is a curve through $(m,\boldsymbol{v}) \in E$ with tangent vector $\sigma_\mathfrak{g}(m,\boldsymbol{v})$ at $t=0$. So
    \begin{align}
    \begin{split}
    \pi_* \circ \sigma_\mathfrak{g}(m,\boldsymbol{v}) &= \frac{d}{dt}\Big|_{t=0}\pi \circ \Sigma_{\text{Exp}(t\mathfrak{g})}(m,\boldsymbol{v}) \\
    &= \frac{d}{dt}\Big|_{t=0}\tilde \Sigma_{\text{Exp}(t\mathfrak{g})}\pi (m,\boldsymbol{v}) \\
    &= \tilde \sigma_{\mathfrak{g}}\circ \pi (m,\boldsymbol{v}).
    \end{split}
    \end{align}
\end{proof}

$\mathfrak{X}(E)$ is a Lie algebra with respect to the Jacobi-Lie bracket
\begin{equation}
    [X,Y]_{-} = -(XY-YX).
\end{equation}
The minus sign in this definition is needed for the next result and accounts for the fact that $\Sigma$ and $\sigma$ act on the left rather than on the right. The following proposition justifies calling $\sigma$ a Lie algebra bundle representation.
\begin{proposition}
$\sigma:\emph{Lie}(G)\rightarrow \mathfrak{X}(E)$ is a Lie algebra homomorphism. 
\end{proposition}
\begin{proof}
    See Ref. \cite{lee_smooth_manifolds}, Theorem 20.18 (a).
\end{proof}

Following the notation of Weinberg \cite{Weinberg1995}, the infinitesimal Poincar\'{e} transformations with homogeneous and translational parts
\begin{align}
    {\Lambda^\mu}_\nu &= {\delta^\mu}_\nu + {\omega^\mu}_\nu \\
    \mathsf{a}^\mu &= \epsilon^\mu 
\end{align}
define the generators of the Poincar\'{e} group $iJ^{\mu \nu}$ and $iP^\mu$ by 
\begin{equation}\label{eq:Poincar\'{e}_generators}
    \Sigma(I+\omega,\epsilon) = 1 + \frac{1}{2} \omega_{\sigma \eta} (iJ^{\sigma \eta}) - \epsilon_\rho (iP^\rho) + O(\omega^2, \epsilon^2, \omega \epsilon)
\end{equation}
where $\omega_{\sigma\eta}$ is anti-symmetric. The generators are $\omega$- and $\epsilon$-independent and go by standard names: the Hamiltonian $H = P^0$ generating time translation, momentum $\boldsymbol{P}=\{P^1,P^2,P^3\}$ generating spatial translations, angular momentum $\boldsymbol{J}=\{J^{23},J^{31},J^{12}\}\doteq \{J_1,J_2,J_3\}$ generating spatial rotations, and boost $\boldsymbol{K}=\{J^{01},J^{02},J^{03}\} \doteq \{K_1,K_2,K_3\}$ generating the boost transformations. Given $\boldsymbol{v}\in \mathbb{R}^3$, these also define the generators $J_{\boldsymbol{v}}=\boldsymbol{v}\cdot \boldsymbol{J}$ and $K_{\boldsymbol{v}}=\boldsymbol{v}\cdot \boldsymbol{K}$. The generators satisfy the commutation relations
\begin{align}
    [J_a,J_b] &= i\epsilon_{abc}J_c \,, \;\;  [J_a, K_b] = i\epsilon_{abc} K_c \label{eq:comm_1}\\
    [K_a,K_b] &= -i\epsilon_{abc}J_c \,, \;\; [J_a,P_b] = i\epsilon_{abc}P_c \label{eq:comm_2}\\
    [K_a,P_b] &= iH\delta_{ab} \,, \;\; [K_a,H] = iP_a \label{eq:comm_3}\\
    [J_a,H] &= [P_a,H] = [P_a,P_b] = [H,H] = 0 \label{eq:comm_4}.
\end{align}

%%%%%%%%%%%%%%% Singularities in massless vector space representations %%%%%%%%%%%%%%
\subsection{Singularities in massless vector space representations}\label{subsec:Little_group_singularity}
In this section we illustrate how the conventional form of the little group method \cite{Wigner1939,Weinberg1995} produces vector space representations of the Poincar\'{e} group with singularities. We will discuss this method in its modern form, as presented by Weinberg \cite{Weinberg1995}, and in its original and slightly more rigorous form as developed by Wigner \cite{Wigner1939}. In the modern formulation, one constructs unitary representations $\Sigma$ of the Poincar\'{e} group on the space of single-particle states, that is, on the (generalized) eigenvectors $\Psi_{k,a}$ of the momentum operator:
\begin{equation}
    P^\mu\Psi_{k,a} = k^\mu \Psi_{k,a}.
\end{equation}
Here, $a$ labels the internal degrees of freedom. It is assumed that $a$ is a discrete index, reflecting the fact that no particles with continuous internal degrees of freedom have been experimentally observed \cite{Weinberg1995}. The little group method seeks to construct all possible representations $\Sigma$ as follows. One begins by fixing some momentum $\bar{k}$, and considering the little group $\LittleGroup{\bar{k}}$ consisting of the elements in $\mathrm{SO}^+(3,1)$ stabilizing $\bar{k}$:
\begin{equation}
    \LittleGroup{\bar k} \doteq \{\Lambda\in \mathrm{SO}^+(3,1)|\Lambda (\bar k)=\bar k\}.
\end{equation}
The little groups for different choices of $\bar{k}$ are conjugate to each other since if $\Lambda$ is any Poincar\'{e} transformation taking $k_1$ to $k_2$, then
\begin{equation}
    \LittleGroup{k_2} = \Lambda \LittleGroup{k_1} \Lambda^{-1}.
\end{equation}
Thus, up to isomorphism, $\LittleGroup{\bar{k}}$ is independent of $\bar{k}$. The single-particle states at $\bar{k}$ form a representation of the little group, since for $\Lambda \in \LittleGroup{\bar{k}}$
\begin{equation}\label{eq:little_group_action}
    \Sigma(\Lambda)\Psi_{\bar{k},a} = \sum_b D_{ab}(\Lambda)\Psi_{\bar{k},j}
\end{equation}
for some scalars $D_{ab}(\Lambda)$. It is much easier to classify the little group representations since the little group is smaller than the Poincar\'{e} group. For massive particles, the momentum space is a mass hyperboloid, and one can choose the reference momentum $\bar{k} = (1,0,0,0)$ as that of a particle at rest. The little group is then clearly $\mathrm{SO}(3)$. The finite-dimensional representations are $2s+1$ dimensional and labeled by the spin $s$. For non-projective representations, s is restricted to positive integer values. The situation is considerably different for massless particles. In this case the momenta are lightlike and there is no preferred reference momentum as there is no rest frame. One typically chooses the reference momentum $k = (1,0,0,1)$. One can show that the little group is $\mathrm{ISO}(2)$ \cite{Weinberg1995}. In contrast to $\mathrm{SO}(3)$, the finite-dimensional representations of $\mathrm{ISO}(2)$ are all one-dimensional. They are labeled not by spin but by the helicity $h$ which, for non-projective representations, can take on all integer values. The change in the momentum space at $m=0$ accounts for the fact that massless particles such as photons are characterized by helicity rather than by spin.

To construct representations of $\mathrm{SO}^+(3,1)$ from representations of the little group, one must relate the states with momentum $\bar{k}$ to those with other momenta. The conventional method for doing so is to assign to each $k$ in the momentum space $M$ a Lorentz transformation $L(k)$ such that
\begin{equation}\label{eq:L_function}
    L(k)\bar{k} = k.
\end{equation}
Thus, $L:M\rightarrow \mathrm{SO}^+(3,1)$. Up to normalization, one then \emph{defines} the single-particle states of momentum $k$ in terms of those of the reference momentum $\bar k$ by \cite{Weinberg1995}
\begin{equation}\label{eq:single_particle_space}
    \Psi_{k,a} = \Sigma(L(k))\Psi_{\bar k, a}.
\end{equation}
Then, the action of an arbitrary $\Lambda \in \mathrm{SO}^+(3,1)$ is given by
\begin{equation}\label{eq:arb_action}
    \Sigma(\Lambda)\Psi_{k,a} = \Sigma(L(\Lambda k)) \Sigma(L(\Lambda k)^{-1}\Lambda L(k))\Psi_{\bar k,a}.
\end{equation}
Since $L(\Lambda k)^{-1}\Lambda L(k) \in \LittleGroup{\bar{k}}$, the action in Eq.\,(\ref{eq:arb_action}) is completely determined by Eqs.\,(\ref{eq:little_group_action}) and  (\ref{eq:single_particle_space}). Equation (\ref{eq:arb_action}) thus extends the little group action to a Poincar\'{e} action on all single-particle states. However, the validity of this procedure relies on being able to smoothly define $L:M\rightarrow \mathrm{SO}^+{(3,1)}$ satisfying Eq.\,(\ref{eq:L_function}). If $L$ is not smooth, then the space of single-particle states constructed in Eq.\,(\ref{eq:single_particle_space}) does not have a smooth structure. It is then not possible to say that the Poincar\'{e} group acts smoothly on the single-particle states, and therefore the formal representations described by Eq.\,(\ref{eq:arb_action}) are not actually Lie group representations.

For massive particles this is not an issue. For the reference momentum $\bar k = (1,0,0,0)$, one can smoothly choose $L(k)$ to be the unique boost taking $\bar k$ to $k$. However, for massless particles we prove there are no smooth choices of $L$.

\begin{theorem}\label{thm:no_go}
    Let $\bar{k} \in \mathcal{L}_+$. There exists no smooth function $L:\mathcal{L}_+ \rightarrow \mathrm{SO}^+(3,1)$ satisfying Eq.\,(\ref{eq:L_function}) for all $k$. 
\end{theorem}

\begin{proof}
    Suppose there were such a function $L(k)$. We will prove in Theorem \ref{thm:RL_equivariant} that $R$- and $L$-polarization states are preserved under the action $\Sigma$ of $\mathrm{SO}^+(3,1)$. Additionally, a photon will appear as a propagating wave in any reference frame, so under a Lorentz transformation a wave $(k,\boldsymbol{E}) \in \gamma_k$ with $\boldsymbol{E} \neq 0$ will transform into another wave with nonzero electric field. This can also be seen by taking the dot and cross product of Eq. (\ref{eq:boost}) with $\boldsymbol{v}$. If one fixes a choice of $(\bar{k},\boldsymbol{E}_+) \in \gamma_+(\bar{k})$ with $\boldsymbol{E}_+ \neq 0$ then $\Sigma_{L(k)}(\bar{k}, \boldsymbol{E}_+)$, considered as a function of $k$, is a continuous nonvanishing section of $\gamma_+$, contradicting the fact that $\gamma_+$ is a nontrivial bundle. Thus, no continuous function $L(k)$ exists.
\end{proof}

As an example, Weinberg \cite{Weinberg1995} chooses
\begin{equation}
    L(k) = \exp(i\phi J_3) \exp(i \theta J_2) B(|\boldsymbol{k}|)
\end{equation}
where
\begin{equation}
    k = |\boldsymbol{k}| (1,\sin \theta \, \cos \phi, \sin \theta \, \sin \phi, \cos \theta)
\end{equation}
and $B$ is a boost in the $z$ direction. However, $L(k)$ is discontinuous 
at $\theta = 0$ and $\theta = \pi$ since
\begin{align}
\begin{split}
    &\lim_{\theta\rightarrow 0} \exp(i\phi J_3) \exp(i \theta J_2) = \exp(i\phi J_3), \\
    &\lim_{\theta\rightarrow \pi} \exp(i\phi J_3) \exp(i \theta J_2) = -\exp(i\phi J_3)
\end{split}
\end{align}
both depend on the value of $\phi$. This illustrates that the modern formulation of the little group method \cite{Weinberg1995,Maggiore2005} fails for massless particles. 

This issue presents slightly differently in Wigner's original and more rigorous formulation on the little group method \cite{Wigner1939}. The single-particle states used by Weinberg do not themselves form a Hilbert space (for example, they are are delta-function normalized \cite{Weinberg1995}). Weinberg instead builds representations on the Hilbert space of $L^2$ wavefunctions $\psi(k,a)$ where again $a$ is a finite index. Given a unitary matrix representation of the little group $H_{\bar{k}}$ in which $\Lambda \mapsto D_{ab}(\Lambda)$, Wigner defines the Hilbert space representation
\begin{equation}\label{eq:Wigner_rep}
    [\Sigma(\Lambda)\psi](k,a) = \sum_b D_{ab}\Big(L(k)^{-1}\Lambda L(\Lambda^{-1}k) \Big) \psi(\Lambda^{-1}k,b).
\end{equation}
In the massless case, $L$ is discontinuous for at least one $k'$. Thus, if $\psi$ is smooth, the transformed wavefunction $\Sigma(\Lambda)\psi$ is discontinuous at $k'$ and $\Lambda k'$. The trivial exceptions to this are when $\Lambda$ is the identity or when $D_{ab}$ is the trivial representation with $D_{ab}(\Lambda)=\mathds{1}$ for every $\Lambda$. The latter is interesting because it shows that the massless helicity $0$ representation is smooth. However, all the nonzero helicity representations are not smooth. We do note that such discontinuous wavefunctions are still elements of the $L^2$ Hilbert space, and Wigner showed that these representations can be regarded as continuous in the technical sense that if $\Lambda_\alpha \rightarrow \Lambda$, then $|\Lambda_\alpha \psi| \rightarrow |\Lambda \psi|$ in the $L^2$ norm. Nevertheless, such non-smooth representations are both practically awkward and physically unnatural.

One way to understand this issue is that the conventional little group method constructs representations on the ``wrong'' Hilbert space for massless particles. From a mathematical standpoint, there is a single infinite-dimensional Hilbert space since between any two Hilbert spaces $\mathcal{H}_1$ and $\mathcal{H}_2$, there exists a (non-unique) unitary isomorphism $F:\mathcal{H}_1\rightarrow \mathcal{H}_2$. However, there are many concrete manifestations of this Hilbert space, for example, the spaces $L^2(\mathbb{R}^n)$ of $\Comp$-valued square-integrable functions on $\mathbb{R}^n$. For a massive particle with spin $s$, Eq. (\ref{eq:Wigner_rep}) gives a smooth representation on the $2s+1$ component wavefunctions over the mass hyperboloid, that is, on the Hilbert space $\mathcal{H} = \bigoplus_{2s+1} L^2(\mathbb{R}^3)$. However, for any other Hilbert space, say $L^2(\mathbb{R})$, there exists a unitary isomorphism $F:\mathcal{H} \rightarrow L^2(\mathbb{R})$, which then induces a representation on $L^2(\mathbb{R})$. Of course, there is no guarantee that $F$ will map smooth wavefunctions into smooth wavefunctions, and thus, the induced representation on $L^2(\mathbb{R})$ is generally not smooth. This is to say, it is possible to describe a spin $s$ particle by a single wavefunction over the real line, but it would behave pathologically under Poincar\'{e} transformations, and is thus unnatural from a physical standpoint. Indeed, such representations are never used in physics. In this sense, one might regard $L^2(\mathbb{R})$ as the ``wrong'' Hilbert space to represent spin $s$ particles. That the conventional little group method produces non-smooth massless representation on $L^2(\mathcal{L}_+)$ suggests that $L^2(\mathcal{L}_+)$ is not a well-suited Hilbert space for massless particles (except when $h=0$). We will show in section \ref{subsec:sectional_reps} that $L^2$ sections of vector bundles over $\mathcal{L}_+$ support smooth representations, and thus form natural Hilbert spaces for massless particles.

%%%%%%%%%%%% Vector bundle little group %%%%%%%%%
\subsection{The little group method for massless vector bundle representations}\label{subsec:bundle_little_group}

We can resolve global non-smoothness issues in the massless case by considering vector bundle representations of the Poincar\'{e} group over the lightcone. We will show that a vector bundle version of the little group method can be used to canonically decompose any unitary $\mathrm{ISO}^+(3,1)$-equivariant vector bundle $\pi:E\rightarrow \mathcal{L}_+$ into irreducible bundle representations labeled by helicity.

We begin by defining bundle representations of the little group.
\begin{definition}[Stabilizing vector bundle representation]
    A $G$-equivariant vector bundle $\pi:E\rightarrow M$ with group action $(\Sigma,\tilde \Sigma)$ is said to be stabilizing if $\tilde \Sigma_g$ is the identity for every $g\in G$, that is, if
    \begin{equation}
        \Sigma_g(E_k) \seq E_k
    \end{equation}
    for every $k \in M$.
\end{definition}
\begin{definition}[Little group of a vector bundle representation]
    Let $\pi:E\rightarrow M$ be a homogeneous $G$-equivariant vector bundle. The little group at $k \in M$ is defined by
    \begin{equation}
        \LittleGroup{k} = \{g \in G | \tilde \Sigma_g(k) = k\}.
    \end{equation}    
    The little group $H$ of the representation is defined as the isomorphism type of $\LittleGroup{k}$, which is independent of $k$ because $\LittleGroup{\Tilde{\Sigma}_g k} = g\LittleGroup{k} g^{-1}$ and $E$ is homogeneous. 
    
    In the case that $G$ is the Poincar\'{e} group, $H$ and $\LittleGroup{k}$ are restricted to elements of $\mathrm{SO}^+(3,1)$, that is, they are defined in terms of the corresponding proper orthochronous Lorentz action.
\end{definition}

\begin{theorem}[Little group representation]\label{thm:little_group_rep}
    Let $\pi:E\rightarrow M$ be a homogeneous $G$-equivariant vector bundle with action $(\Sigma, \tilde \Sigma)$ and little group H. Let $\LittleGroup{k}$ denote the little group at $k\in M$. Suppose $f:H\times M \rightarrow G$ is a smooth map such that $f(h,k) \in \LittleGroup{k}$ for all $h,k$, and for each fixed $k$, $f(h,k)$ is a group homomorphism from $H$ to $G$. Such an $f$ induces a stabilizing bundle representation of $H$ on $E$ with action $\Pi$ given by
    \begin{equation}\label{eq:little_group_rep_general}
        \Pi_h(k,v) = \Sigma_{f(h,k)}(k,v)
    \end{equation}
    for $(k,v) \in E$. $\Pi$ is unitary if $\Sigma$ is unitary. A bundle representation of H induced by such an $f$ is called a little group representation. 
\end{theorem}
\begin{proof}
    $\Pi$ is smooth because $\Sigma$ and $f$ are. It is a vector bundle representation because
    \begin{align}
        \Pi_{h_1h_2}(k,v) &= \Sigma_{f(h_1 h_2,k)}(k,v) = \Sigma_{f(h_1,k)f(h_2,k)}(k,v) \\ &= \Sigma_{f(h_1,k)}\Sigma_{f(h_2,k)}(k,v) = \Pi_{h_1}\Pi_{h_2}(k,v).
    \end{align}
    It is stabilizing since $f(h,k) \in \LittleGroup{k}$. Since $\Pi$ can be written in terms of $\Sigma$, if $\Sigma$ is unitary then so is $\Pi$.
\end{proof}
By the definition of the little group $H$, for each $k$ there exists an isomorphism $f_k:H\rightarrow \LittleGroup{k}$. Note that $f_k$ is not typically unique. One can then define
\begin{equation}
    f(h,k) = f_k(h).
\end{equation}
However, this $f$ is not generally smooth since there is no guarantee that the choices of $f_k$ fit together smoothly. Our first goal is construct a canonical smooth $f$ for an arbitrary $\mathrm{ISO}^+(3,1)$-equivariant vector bundle $\pi:E\rightarrow\mathcal{L}_+$, which will then furnish a canonical little group representation. Note that any such bundle $E$ is homogeneous since $\mathrm{ISO}^+(3,1)$ acts transitively on $\mathcal{L}_+$.

For a fixed $\bar{k} \in \mathcal{L}_+$, any $\Lambda \in \LittleGroup{\bar{k}}$ satisfies $\Sigma(\Lambda)E_{\bar k} \subseteq E_{\bar k}$, so $\Sigma$ gives a finite-dimensional vector space representation of $\LittleGroup{\bar k}\cong \mathrm{ISO}(2)$ on the fiber $\gamma_{\bar{k}}$; we will use $ \Sigma^{\bar k}$ to denote this restriction of $\Sigma$ to $E_{\bar k}$.

Consider first the special case of $\bar{k} = (1,0,0,1)$. Weinberg \cite{Weinberg1995} showed that the little group of this $\bar{k}$ can be described by a function of three parameters $W^{\bar{k}}(\theta, \alpha, \beta):\mathrm{ISO}(2)\rightarrow \mathrm{ISO}^+(3,1)$ and is represented by
\begin{align}\label{eq:iso_rep_k_vary}
\begin{split}
    {\Sigma}(W^{\bar{k}}(\theta, \alpha, \beta)) &= \exp{(i \alpha A + i\beta B + i\theta J_3)}, \\
    A&=J_2 + K_1, \\
    B&=-J_1+K_2.
\end{split}
\end{align}
Under the isomorphism with $\mathrm{ISO}(2)$, $J_3$ generates 2D rotations while $A$ and $B$ generate translations. They satisfy the commutation relations
\begin{align}
    [J_3,A] &= +iB, \label{eq:little_group_comm_1}\\
    [J_3,B] &= -iA,\label{eq:little_group_comm_2} \\
    [A,B] &= 0 \label{eq:little_group_comm_3}.
\end{align}
It is known that all finite-dimensional irreducible representations of $\mathrm{ISO}(2)$ are one-dimensional. The general case was proved by Schwarz \cite{Schwarz1971}; simpler proofs for unitary representations of $\mathrm{ISO}(2)$ are presented by Weinberg \cite{Weinberg1995} and Maggiore \cite{Maggiore2005}. Therefore, on each such irreducible representation,  $J_3$, $A$, and $B$ must all be multiplication operators and thus commute. Eqs. (\ref{eq:little_group_comm_1}) and (\ref{eq:little_group_comm_2}) then imply $A=B=0$ on each irreducible representation, which in turn implies $A=B=0$ on all of $E_{\bar k}$. Thus, 
\begin{align}
    J_2 + K_1 = 0 \label{eq:JK_1}, \\
    -J_1 + K_2 = 0\label{eq:JK_2}
\end{align}
when restricted to $\gamma_{\bar k}$. We note two subtleties here. First, the relations (\ref{eq:JK_1}) and (\ref{eq:JK_2}) hold only on the fiber $E_{\bar k}$. Second, $J_2$ and $K_1$ are not independently operators in the little group representation since neither generate transformations leaving $\bar{k}$ invariant. However, since $J_2 + K_1 = 0$ in the little group representation, this relation holds also in the bundle representation when restricted to $E_{\bar{k}}$, and in this sense the perpendicular boosts and angular momenta are related by $J_2 = -K_1$ and $J_1 = K_2$ on $E_{\bar{k}}$.

Since $A=B=0$, the irreducible representations are completely determined by the action of the generator $J_3$. In fact, they are just the eigenspaces of $J_3$. Let $\Psi_{\bar k,h}$ be the eigenvectors with eigenvalues $h$:
\begin{equation}
    J_3\Psi_{\bar k,h} = h \Psi_{\bar k, \sigma}. 
\end{equation}
$h$ defines the helicity of each irreducible representation of $\LittleGroup{\bar{k}}$. Since
\begin{equation}
    e^{i2\pi J_3}\Psi_{\bar k,h} = e^{2\pi i h}\Psi_{\bar k,h} = \Psi_{\bar k,h},
\end{equation}
$h$ must be an integer. Note that if one allows projective representations, $h$ may also be a half-integer \cite{Weinberg1995}.

These results generalize easily to arbitrary $k = (|k|,\boldsymbol{k}) \in \mathcal{L}_+$ as there is nothing special about $\bar{k} = (1,0,0,1)$. Choose any $\boldsymbol{f}_1,\boldsymbol{f}_2 \in \mathbb{R}^3$ such that $(\boldsymbol{f}_1, \boldsymbol{f}_2, \boldsymbol{\hat k})$ form a right-handed orthonormal coordinate system. Then the little group $\LittleGroup{k}$ is given by a function $W^k(\theta, \alpha, \beta)$ represented by
\begin{align}\label{eq:iso_rep_k_fixed}
    \Sigma(W^k(\theta, \alpha, \beta)) &= \exp{[i \alpha A^k + i\beta B^k + i\theta (\boldsymbol{\hat k} \cdot \boldsymbol{J})]}, \\
    A^k &= \boldsymbol{f}_2 \cdot \boldsymbol{J} + \boldsymbol{f}_1 \cdot \boldsymbol{K} = 0 \label{eq:JK_general_1}, \\
    B^k &= - \boldsymbol{f}_1 \cdot \boldsymbol{J} + \boldsymbol{f}_2 \cdot \boldsymbol{K} = 0 \label{eq:JK_general_2}.
\end{align}
The latter two equations say that the boosts and perpendicular angular momentum are related on $E_k$ by
\begin{equation} \label{eq:JK}
    \boldsymbol{J}_\perp = -\boldsymbol{\hat k}\times \boldsymbol{K}_\perp = -\boldsymbol{\hat k}\times \boldsymbol{K}
\end{equation}
where
\begin{align}
    \boldsymbol{J}_\perp &\doteq \boldsymbol{J} - \boldsymbol{\hat k}(\boldsymbol{\hat k} \cdot \boldsymbol{J}) \label{eq:J_perp}, \\
    \boldsymbol{K}_\perp &\doteq \boldsymbol{K} - \boldsymbol{\hat k}(\boldsymbol{\hat k} \cdot \boldsymbol{K}) \label{eq:k_perp}.
\end{align}
The only nontrivial generator of the little group is $\boldsymbol{\hat k} \cdot \boldsymbol{J}$, and its eigenspaces are again labeled by integer helicities. For systems with a well-defined spatial-inversion symmetry, such as the photon system, it can further be shown that the eigenvalues must come in $\pm h$ pairs \cite{Weinberg1995}.

We would like to show that the fiber-wise little group representations fit together to form a bundle representation of the little group  $\mathrm{ISO}(2)$ via Theorem \ref{thm:little_group_rep}. However, $W^k(\theta,\alpha,\beta)$ considered as a function from $\mathrm{ISO}(2) \times \mathcal{L}_+$ to $\mathrm{ISO}^+(3,1)$ is not smooth because it is not possible to smoothly choose $(\boldsymbol{f}_{1}(k),\boldsymbol{f}_2(k))$ by the hairy ball theorem. However, the fiber-wise little group representations have the important property that the translations $A^k$ and $B^k$ act trivially. Thus, $f:\mathrm{ISO}(2) \times \mathcal{L}_+ \rightarrow \mathrm{SO}^+(3,1)$ given by \begin{equation}
    f\big((\theta,\alpha,\beta),k\big) = W^k(\theta,0,0)
\end{equation}
is smooth and produces precisely the same action as $W^k(\theta,\alpha,\beta)$ under $\Sigma$:
\begin{equation}\label{eq:f_little_group}
    \Sigma\big(f((\theta,\alpha,\beta),k)\big) = \Sigma(W^k(\theta,\alpha,\beta)).
\end{equation}
Applying Theorem \ref{thm:little_group_rep} to this $f$ then gives a canonical little group action.
\begin{theorem}\label{thm:associated_little_group}
    Every $\mathrm{ISO}^+(3,1)$-equivariant vector bundle $\pi:E\rightarrow \mathcal{L}_+$ with group action $\Sigma$ has an associated little group representation of $\mathrm{ISO}(2)$ with action $\Pi$ given by
    \begin{equation}
        \Pi(\theta,\alpha,\beta)(k,v) = (k,e^{i\theta \chi}v) = (k,e^{i\theta (\boldsymbol{\hat k} \cdot \boldsymbol{J})}v)
    \end{equation}
    where
    \begin{equation}
        \chi = \boldsymbol{\hat{k}} \cdot \boldsymbol{J}
    \end{equation}
    is the helicity operator. If $\Sigma$ is unitary, then so is $\Pi$.
\end{theorem}

Note that the helicity operator can be written as
\begin{equation}\label{eq:helicity}
    \chi = \boldsymbol{\hat{k}} \cdot \boldsymbol{J} = \frac{\boldsymbol{P}}{|\boldsymbol{k}|} \cdot \boldsymbol{J} = \frac{1}{H} \boldsymbol{P} \cdot \boldsymbol{J}.
\end{equation}
This operator is smooth on the vector bundle since $H = |\boldsymbol{k}|\neq 0$ on the lightcone and is thus invertible. Note that $H^{-1}$ commutes with $\boldsymbol{P}$ and $\boldsymbol{J}$. Furthermore, for each $a\in \{1,2,3\}$, $P_a$ and $J_a$ commute, so $\boldsymbol{P} \cdot \boldsymbol{J} = \boldsymbol{J} \cdot \boldsymbol{P}$. Thus, the terms in Eq.\,(\ref{eq:helicity}) can be rearranged, and there is no ordering ambiguity in the definition of $\chi$. The following is an important property of $\chi$.
\begin{theorem}
    The $\mathrm{ISO}^+(3,1)$ action $\Sigma$ and its associated little group action $\Pi$ commute.
\end{theorem}
\begin{proof}
    We prove this by showing that $\chi$ commutes with all generators of the $\mathrm{ISO}^+(3,1)$ action. In the following calculations, we implicitly sum over all repeated indices. That $\chi$ commutes with the spacetime translation generators and the Hamiltonian $H$ is trivial. For the generators $P_b$, we have
    \begin{align}
    \begin{split}
        [\chi, P_b] &= [\frac{P_a}{H}J_a,P_b] = \frac{P_a}{H}[J_a,P_b] \\
        &= i\epsilon_{abc}H^{-1}P_aP_c 
        = 0.
    \end{split}
    \end{align}
    For $J_b$,
    \begin{align}
    \begin{split}
        [\chi,J_b] &= H^{-1}[P_aJ_a,J_b] \\
        &= H^{-1}(P_a[J_a,J_b] + [P_a,J_b]J_a) \\
        &= iH^{-1}(\epsilon_{abc}P_aJ_c + \epsilon_{abc}P_cJ_a) \\
        &= 0.
    \end{split}
    \end{align}
    To show that $\chi$ commutes with boosts, note that
    \begin{align}
    \begin{split}
        [H^{-1},K_b] &= H^{-1}K_b - K_bH^{-1} \\
        &= H^{-1}K_bHH^{-1} - H^{-1}HK_b H^{-1}  \\
        &= H^{-1}[K_b,H]H^{-1} = iH^{-1}P_bH^{-1} \\
        &= iH^{-2}P_b.
    \end{split}
    \end{align}
    Thus,
    \begin{align}
    \begin{split}
        [\chi, K_b] &= [H^{-1}P_aJ_a,K_b] \\
        &= H^{-1}P_a[J_a,K_b] + [H^{-1}P_a,K_b]J_a \\
        &= i\epsilon_{abc}H^{-1}P_aK_c + H^{-1}[P_a,K_b]J_a + [H^{-1},K_b]P_aJ_a \\
        &= -iH^{-1}(\boldsymbol{P}\times \boldsymbol{K})_b - i\delta_{ab}J_a + iH^{-2}P_bP_aJ_a \\
        &= i(-J_b - (\boldsymbol{\hat k} \times\boldsymbol{K})_b + \hat{k}_b(\boldsymbol{k}\cdot \boldsymbol{J}))
    \end{split}
    \end{align}    
    By the relation (\ref{eq:JK}) between the boost and rotation generators for massless representations, these terms cancel giving
    \begin{equation}
        [\chi, K_b] = 0.
    \end{equation}
\end{proof}
We can now show that the massless unitary irreducible bundle representations of $\mathrm{ISO}^+(3,1)$ are the constant helicity representations.
\begin{theorem}\label{thm:bundle_decomp}
    Let $\pi:E\rightarrow \mathcal{L}_+$ be a unitary $\mathrm{ISO}^+(3,1)$-equivariant vector bundle of rank $r$. Then $E$ can be decomposed as
    \begin{equation}
        E = E_1 \oplus \cdots \oplus E_r
    \end{equation}
    where the $E_j$ are unitary irreducible $\mathrm{ISO}^+(3,1)$-equivariant line subbundles of $E$. Each $E_j$ has definite helicity $h_j$ in the sense that every element of $E_j$ is an eigenvector of $\chi$ with helicity $h_j$.
\end{theorem}

\begin{proof}
    Choose some $\bar k \in \mathcal{L}_+$. Let $(v_1,...,v_r)$ be a basis of the fiber $E_{\bar k}$ consisting of eigenvectors of $\chi$ with helicities $(h_1,...,h_r)$. Define the subset $E_j \seq E$ as the orbit of $(\bar{k},v_j)$ under the group action and scalar multiplication:
    \begin{equation}\label{eq:def_E_j}
        E_j \doteq \{c\Sigma(L)(\bar{k},v_j)|L\in \mathrm{SO}^+(3,1), c \in \Comp\}.
    \end{equation}
    We will show that each $E_j$ is a rank-$1$ subbundle of $E$ by the subbundle criterion. Partition $E_j$ as
    \begin{equation}
        E_j = \amalg_{k\in \mathcal{L}_+}E_j(k)
    \end{equation}
    where $E_j(k)$ is the subset of $E_j$ consisting of vectors at $k$. Each $E_j(k)$ is a vector space and it is nonempty since if $L(k)$ is any Lorentz transformation taking $\bar{k}$ to $k$, then $\Sigma(L(k))(\bar{k},v_j) \in E_{j}(k)$. Thus, $E_j(k)$ is at least one-dimensional. Suppose $(k,w_1)$ and $(k,w_2)$ are both in $E_j(k)$ and nonzero. Then
    \begin{align}
        (k,w_1) &= c_1 \Sigma(L_1)(\bar{k},v_j) \\ 
        (k,w_2) &= c_2 \Sigma(L_2)(\bar{k},v_j)
    \end{align}
    for some nonzero scalars $(c_1,c_2)$ and Lorentz transformations $(L_1, L_2)$. As $L_2^{-1}L_1$ is in the little group $\LittleGroup{\bar{k}}$ of $\bar{k}$,
    \begin{align}
        \Sigma(L_2^{-1}L_1)(\bar{k},v_j) = (\bar{k},e^{i\theta h_j}v_j)
    \end{align}
    for some $\theta$. Then
    \begin{align}
    \begin{split}
        (k,w_1) &= c_1\Sigma(L_2)\Sigma(L_2^{-1}L_1)(\bar k, v_j) \\
        &= c_1c_2^{-1}e^{i\theta h_j}(k,w_2)
    \end{split}
    \end{align}
    so $(k,w_1)$ and $(k,w_2)$ are linearly dependent. This shows that every $E_j(k)$ is one-dimensional. 
    
    Now, let $k_0 \in \mathcal{L}_+$ be arbitrary and $U$ be a small ball about $k_0$. Choose the radius of $U$ to be small enough that it does not enclose the origin. We construct a smooth function $L:U\rightarrow \mathrm{SO}^+(3,1)$ such that $L(k)k_0=k$. Note that we showed in Theorem \ref{thm:no_go} that it is not possible to construct such a function if the domain is all of $\mathcal{L}_+$. However, it is possible to construct such a function locally. Indeed, we can simply choose $L(k) = R(k)B(k)$ where $B(k)$ is the boost parallel to $\boldsymbol{k}_0$ taking $\boldsymbol{k}_0$ to $\boldsymbol{k}_0\frac{|\boldsymbol{k}|}{|\boldsymbol{k}_0|}$ and $R(k)$ is the unique rotation by angle $0 \leq \theta < \pi$ along the great circle connecting $\boldsymbol{k}_0\frac{|\boldsymbol{k}|}{|\boldsymbol{k}_0|}$ to $\boldsymbol{k}$. Let $(k_0,v)$ be a nonzero vector in $E_j(k_0)$. Then $\Sigma(L(k))(k_0,v)$ is a smooth section of $E$ over $U$ that span $E_j(k)$ for each $k \in U$. Thus, each $E_j$ is a line subbundle of $E$ by Lemma \ref{lm:subbundle_criterion}. They are unitary equivariant subbundles by their definitions in Eq.\,(\ref{eq:def_E_j}). Every vector $(k,w_1) \in E_1$ is an eigenvector of $\chi$ with helicity $h_j$ since by Theorem $\ref{thm:associated_little_group}$,
    \begin{align}
    \begin{split}
        \chi (k,w_1) &= \chi c_1 \Sigma(L_1)(\bar{k},v_j) \\
        &= c_1 \Sigma(L_1)\chi (\bar{k},v_j) \\&= h_j (k,w_1).
    \end{split}
    \end{align}
\end{proof}
In this bundle decomposition, it is possible for multiple line bundles $E_j$ to have the same helicity. The next result says that such representations are equivalent, showing that the irreducible bundle representations of $\mathrm{ISO}^+(3,1)$ are completely characterized by their helicity.

\begin{theorem}\label{thm:bundle_little_group}
    Suppose $\pi:E_1\rightarrow \mathcal{L}_+$ and $\pi_2:E_2\rightarrow \mathcal{L}_+$ are unitary $\mathrm{ISO}^+(3,1)$-equivariant line bundles with actions $\Sigma_1$ and $\Sigma_2$ and helicities $h_1$ and $h_2$. They are unitarily equivalent representations if and only if $h_1 = h_2$. 
\end{theorem}
\begin{proof}
If $E_1$ and $E_2$ are equivalent representations, then there exists an isomorphism $g:E_1 \rightarrow E_2$ as in Definition \ref{def:equivariant_iso}. Let $\chi_1$ and $\chi_2$ be the helicity operators induced by $\Sigma_1$ and $\Sigma_2$. For $(k,v) \in E_1$, linearity gives
\begin{equation}
    g(e^{i\theta \chi_1 }(k,v)) = g(e^{ih_1 \theta}(k,v)) = e^{i h_1 \theta}g(k,v).
\end{equation}
By equivariance,
\begin{equation}
    g(e^{i\theta \chi_1}(k,v)) = e^{i \theta \chi_2}g(k,v) = e^{ih_2 \theta}g(k,v).
\end{equation}
Thus, $h_1 = h_2$.

Conversely, suppose $h_1 = h_2$. We will construct a unitary isomorphism of representations $g:E_1 \rightarrow E_2$. Fix some $(k_0,v_0) \in E_1$ and $(k_0,w_0) \in E_2$ such that $|v_0|=|w_0|\neq 0$. We define
\begin{equation}\label{eq:g_def}
    g(k_0,v_0) = (k_0,w_0),
\end{equation}
and extend this relation by equivariance and linearity. That is, for every $L \in \mathrm{ISO}^+(3,1)$ and $c \in \mathbb{C}$, define
\begin{align}
    g(\Sigma_1(L)(k_0,v_0)) &= \Sigma_2(L)(k_0,w_0), \label{eq:intertwine} \\
    g(k_0,c v_0) &= c (k_0,w_0). \label{eq:g_linear}
\end{align}
Any $(k,v) \in E_1$ can be expressed as
\begin{equation}\label{eq:kv0}
(k,v) =  c \Sigma_1(\Lambda)(k_0,v_0)   
\end{equation}
for some choice of $c,\Lambda$. Indeed, if one chooses a Lorentz transformation $\Lambda$ such that $\Lambda k_0 = k$, then $\Sigma(\Lambda)(k_0,v_0) \in E_1(k)$. Since $E_1$ is a line bundle, any other vector in $E_1(k)$ can be obtained by scalar multiplication. We must still show that $g$ is single-valued under this definition.
Suppose $\Sigma_1(L_1)(k_0,v_0) = \Sigma_1(L_2)(k_0,v_0)$. We need to show
\begin{equation}
    \Sigma_2(L_1)(k_0,w_0) = \Sigma_2(L_2)(k_0,w_0)
\end{equation}
or equivalently,
\begin{equation}\label{eq:well_defined_1}
    \Sigma_2(L_2^{-1}L_1)(k_0,w_0) = (k_0,w_0).
\end{equation}
Since $L_2^{-1}L_1$ is in the little group of $k_0$, it corresponds to an element $(\theta,\alpha, \beta)$ under the isomorphism of the little group with $\mathrm{ISO}(2)$. Furthermore, since $\Sigma_1(L_2^{-1}L_1)(k_0,v_0)=(k_0,v_0)$ either it corresponds to a pure translation $(0,\alpha,\beta)$ or else $h_1=h_2=0$. In either case, Eq.\,(\ref{eq:well_defined_1}) holds. Next, suppose 
\begin{equation}\label{eq:well_defined_2}
    \Sigma_1(L)(k_0,v_0) = (k_0,cv_0).
\end{equation}
We need to show that
\begin{equation}\label{eq:well_defined_3}
    \Sigma_2(L)(k_0,w_0) = (k_0,c w_0).
\end{equation}
Since $L$ is in the little group of $k_0$, it corresponds to some $(\theta,\alpha,\beta)$ in $\mathrm{ISO}(2)$, and $c = e^{i\theta h_1}$. Furthermore,
\begin{equation}
    \Sigma_2(L)(k_0,w_0) = (k_0,e^{i\theta h_2} w_0).
\end{equation}
Equation (\ref{eq:well_defined_3}) holds since $h_1 = h_2$, showing that $g$ is well-defined.

$g$ is a basepoint preserving unitary vector bundle isomorphism by its defining properties. All that remains to show is that $g$ is equivariant, that is,
\begin{equation}
    g(\Sigma_1(L)(k,v)) = \Sigma_2(L) g(k,v)
\end{equation}
for arbitrary $(k,v) \in E_1$ and $L \in \mathrm{ISO}^+(3,1)$. Choose some $c$ and $\Lambda$ such that Eq.\,(\ref{eq:kv0}) holds. Then, using Eqs.\,(\ref{eq:g_def})-(\ref{eq:g_linear}) we have
\begin{align}
\begin{split}
    g(\Sigma_1(L)(k,v)) &= cg(\Sigma_1(L\Lambda)(k_0,v_0)) \\
        &= c\Sigma_2(L \Lambda)(k_0,w_0) \\
        &= c\Sigma_2(L)\Sigma_2(\Lambda)(k_0,w_0) \\
        &= \Sigma_2(L)g(c \Sigma_1(\Lambda)(k_0,v_0)) \\
        &= \Sigma_2(L)g((k,v)).
\end{split}
\end{align}
Thus, $E_1$ and $E_2$ are equivalent representations.
\end{proof}

We can apply this vector bundle version of the little group method to decompose the photon bundle $\gamma$ into Poincar\'{e} invariant line bundles. First note that $\gamma_\pm$ have definite helicity.
\begin{proposition}\label{prop:RL_helicity}
    The $R$- and $L$-bundles $\gamma_\pm$ have definite helicities $\pm1$. 
\end{proposition}
\begin{proof}
    An arbitrary vector $v_\pm \in \gamma_\pm$ can be written as $v_\pm = \alpha (\boldsymbol{e}_1 \pm i\boldsymbol{e}_2)$ where $(\boldsymbol{e}_1,\boldsymbol{e}_2,\boldsymbol{\hat k})$ is real, right-handed, and orthonormal. $R(\theta) = e^{i(\boldsymbol{k}\cdot \boldsymbol{J})}$ describes a passive rotation by $\theta$ about $\boldsymbol{\hat{k}}$. Using the $(\boldsymbol{e}_1,\boldsymbol{e}_2)$ basis,
    \begin{equation}
        R(\theta)v_\pm = \begin{pmatrix}
            \cos \theta && \sin \theta \\
            -\sin \theta && \cos \theta
        \end{pmatrix}
        \begin{pmatrix}
            1 \\
            \pm i
        \end{pmatrix}
         = e^{\pm i\theta}v_\pm.
    \end{equation}
    Differentiating with respect to $\theta$ gives
    \begin{equation}
        \chi v_\pm = (\boldsymbol{\hat k} \cdot \boldsymbol{J})v_\pm = \pm v_\pm.
    \end{equation}
\end{proof}
We now show that $\gamma_\pm$ are irreducible subrepresentations of the representation $\Sigma$ on $\gamma$, provided we restrict $\Sigma$ from $\mathrm{IO}^+(3,1)$ to $\mathrm{ISO}^{+}(3,1)$. This restriction is necessary since the parity operator $\mathsf{P}$ maps $\gamma_\pm$ to $\gamma_\mp$, and therefore the $R$- and $L$-bundles are not $\mathsf{P}$-symmetric.
\begin{theorem}\label{thm:RL_equivariant}
$\gamma_+$ and $\gamma_-$ are unitary irreducible vector bundle representations of $\mathrm{ISO}^{+}(3,1)$.
\end{theorem}
\begin{proof}
    By Theorem \ref{thm:bundle_decomp}, $\gamma$ decomposes as
    \begin{equation}
        \gamma = E_1 \oplus E_2
    \end{equation}
    into two unitary irreducible line bundle representations of $\mathrm{ISO}^+(3,1)$ with definite helicity. By proposition \ref{prop:RL_helicity}, these line bundles must be $\gamma_\pm$.
\end{proof}
This result shows that $R$- and $L$-photons are globally well-defined. In particular, this construction avoids the singularities that appear in the conventional little group method for massless particles.

\subsection{Vector bundle representations as particles}\label{subsec:sectional_reps}
We have shown that the solutions of Maxwell's equations fit together into a smooth vector bundle $\gamma$. Furthermore, $\gamma$ splits into two unitary irreducible bundle representations of the Poincar\'{e} group, $\gamma_+$ and $\gamma_-$. It is thus natural to consider these bundle representations to be particles. However, by the conventional definition, particles are unitary irreducible Hilbert space representations of the Poincar\'{e} group \cite{Wigner1939, Weinberg1995}. In this section we bridge the gap between these two viewpoints, proving that massless unitary irreducible bundle representations of the Poincar\'{e} group generate corresponding unitary irreducible Hilbert space representations, and can thus be considered particles under the standard definition. These Hilbert space representations are smooth, avoiding the singularities described in section \ref{subsec:Little_group_singularity} that occur in Wigner's little group construction. We note that this method of generating Hilbert space representations from bundle representations was described by Simms \cite{Simms1968}. Our work differs in that we use smooth bundles rather than topological bundles, which allows us to resolve the non-smoothness issues with massless representations. 

A representation of $G$ on the vector bundle $E$ naturally induces a vector space representation of $G$ on the infinite-dimensional vector space $\Gamma(E)$ of smooth sections of E. We refer to this as the sectional representation.
\begin{proposition}\label{prop:sectional_rep}
A $G$-equivariant vector bundle $(E,\Sigma)$ induces a vector space representation of $G$ on $\Gamma(E)$. Given $g\in G$ and a section $\psi:M\rightarrow E$, $g$ is represented by $\Sigma^s_g$, defined by
\begin{equation}\label{eq:sectional_action}
    [\Sigma^s_g \psi](m)=\Sigma_g\big[\psi(\tilde{\Sigma}_{g^{-1}}(m)) \big].
\end{equation}
\end{proposition}

\begin{proof}
    The linearity of the action in $\Gamma(E)$ follows from the fact that $\Sigma_g$ is a bundle map. That the identity $e\in G$ acts trivially on $\Gamma(E)$ follows from $\Sigma_e = I$. If $g_1,g_2\in G$, then
    \begin{align}
    \begin{split}
        [\Sigma^s_{g_1 g_2}\psi](m) &=\Sigma_{g_1 g_2}\big[\psi(\tilde{\Sigma}_{(g_1 g_2)^{-1}}m)\big] \\
        &= \Sigma_{g_1}\Sigma_{g_2}\big[\psi(\tilde{\Sigma}_{g_2^{-1}} \tilde{\Sigma}_{g_1^{-1}}(m))\big] \\
        &= \Sigma_{g_1}\big[(\Sigma^s_{g_2}\psi)(\tilde{\Sigma}_{g_1^{-1}}(m)) \big] \\
        &= \big[\Sigma^s_{g_1}\big(\Sigma^s_{g_2}\psi\big) \big](m)
    \end{split}
    \end{align}
    showing that $\Sigma^s_{g_1 g_2}=\Sigma^s_{g_1}\Sigma^s_{g_2}$. Lastly, $\Sigma^s_g \psi$ is smooth since $\Sigma_g$, $\Sigma_{g^{-1}}$, and $\psi$ are smooth. 
\end{proof}
As an example which will be used in section \ref{subsec:SAM_OAM}, the $\mathrm{ISO}(2)$ little group action on $\gamma$ has a corresponding sectional representation on $\Gamma(\gamma)$, which can be simply expressed due to the splitting $\gamma = \gamma_+ \oplus \gamma_-$. A section $\boldsymbol{E}(\boldsymbol{k}) \in \Gamma(\gamma)$ can be uniquely written in terms of sections $\boldsymbol{E}_\pm$ of $\gamma_\pm$:
\begin{equation}\label{eq:E=E++E-}
    \boldsymbol{E}(\boldsymbol{k}) = \boldsymbol{E}_+(\boldsymbol{k}) + \boldsymbol{E}_-(\boldsymbol{k}).
\end{equation}
That is, there are unique projections $\mathcal{P}_\pm:\Gamma(\gamma) \rightarrow \Gamma(\gamma_\pm)$.
\begin{theorem}\label{thm:sectional_little_group_rep}
    Let $\boldsymbol{E} \in \Gamma(\gamma)$. The sectional representation corresponding to the canonical $\mathrm{ISO}(2)$-little group representation on $\gamma$ is given by
    \begin{equation}
        (\theta, \alpha,\beta)\boldsymbol{E} = e^{i\chi \theta}\boldsymbol{E} = e^{+i\theta}\boldsymbol{E}_+ + e^{-i\theta}\boldsymbol{E}_-. \label{eq:ISO2_sectional_action}
    \end{equation}
    Since the translations $\alpha,\beta$ act trivially, the little group action can be considered as an $\mathrm{SO}(2)\cong S^1$ action:
    \begin{equation}
        \theta \boldsymbol{E} = e^{i\chi \theta}\boldsymbol{E} =  e^{+i\theta}\boldsymbol{E}_+ + e^{-i\theta}\boldsymbol{E}_- \label{eq:SO2_sectional_action}
    \end{equation}
\end{theorem}

\begin{proof}
    Since little group elements act trivially on $k$, the sectional action defined in Eq.\,(\ref{eq:sectional_action}) is given by:
    \begin{align}
        [(\theta,\alpha,\beta)\boldsymbol{E}](\boldsymbol{k}) &= f\big((\theta,\alpha,\beta), \boldsymbol{k}\big)[\boldsymbol{E}(\boldsymbol{k})] = e^{i\chi \theta}\boldsymbol{E}(\boldsymbol{k}) \\
        &= e^{+i\theta}\boldsymbol{E}_+ + e^{-i\theta}\boldsymbol{E}_-.
    \end{align}
\end{proof}

Ideally, if $\Sigma$ is unitary, then the sectional representation $\Sigma^s$ on $\Gamma(E)$ would also be unitary. However, this would imply that the Hermitian structure on $E$ induces a Hermitian product on $\Gamma(E)$, which is not generally true. Suppose then that one additionally specifies a $G$-invariant volume form $d\xi$ on the base manifold $M$, and denote by $\Gamma(E,d\xi)$ the smooth sections which are 
$L^2$-normalizable with respect to the volume form $d\xi$. Then there is an induced Hermitian structure on $\Gamma(E,d\xi)$:
\begin{equation}\label{eq:Hermitian_product_int}
    \langle \psi_1, \psi_2 \rangle = \int_M \langle \psi_1(m),\psi_2(m)\rangle \, d\xi.
\end{equation}
In the present case, the photon bundle has the light cone as its base manifold which has a canonical Lorentz invariant volume form $d\xi = \frac{d^3k}{|\boldsymbol{k}|}$. We will assume this volume form whenever the base manifold of $E$ is $\Lightcone$.
\begin{proposition}\label{prop:unitary_sectional_rep}
    Suppose $\pi :E\rightarrow M$ is a $G$-equivariant vector bundle and $d\xi$ is a $G$-invariant volume form on $M$. If the bundle representation is unitary on $E$, then the sectional representation on $\Gamma(E,d\xi)$ is unitary with respect to the induced Hermitian product (\ref{eq:Hermitian_product_int}).
\end{proposition}
\begin{proof}
For $\psi_1,\psi_2 \in \Gamma(E)$ and $g \in G$, we have
\begin{align}
\begin{split}
\langle \Sigma^s_g \psi_1 ,\Sigma^s_g \psi_2\rangle &= \int_M \langle (\Sigma^s_g \psi_1)(\boldsymbol{k}), (\Sigma^s_g \psi_2)(\boldsymbol{k})  \rangle d\xi \\
&= \int_M \langle \Sigma_g [\psi_1(g^{-1} \boldsymbol{k})], \Sigma_g [\psi_2(g^{-1} \boldsymbol{k})]  \rangle d\xi \\
&= \int_M \langle \psi_1(g^{-1} \boldsymbol{k}), \psi_2(g^{-1} \boldsymbol{k})  \rangle d\xi \\
&= \int_M \langle \psi_1(\boldsymbol{k}), \psi_2(\boldsymbol{k})\rangle g^*(d\xi) \\
&= \langle \psi_1, \psi_2 \rangle
\end{split}
\end{align}
where $g^*(d\xi)=d\xi$ follows from the $G$-invariance of $d\xi$.
\end{proof}

$\Gamma(E,d\xi)$ is not a Hilbert space but rather a pre-Hilbert space since it is not complete. In particular, a sequence of smooth sections can converge in the norm to a discontinuous section. This issue is typical fo $L^2$ spaces, and to remedy it, one can work with the Hilbert space completion of $\Gamma(E,d\xi)$, denoted by $L^2(E,d\xi)$, which is the closure of $\Gamma(E,d\xi)$ in the norm. A unitary sectional action $\Sigma^s$ on $\Gamma(E,d\xi)$ extends to a Hilbert space action on $L^2(E,d\xi)$ in the following way. If $g \in G$ and $\psi \in L^2(E,d\xi)$, then there is a Cauchy sequence $\psi_j$ in $\Gamma(E,d\xi)$ with $\psi_j \rightarrow \psi$. Since $\Sigma^s(g)$ acts unitarily on $\Gamma(E,d\xi)$, $\Sigma^s(g)\psi_j$ is also Cauchy in $\Gamma(E,d\xi)$, and thus converges in $L^2(E,d\xi)$. We can then define
\begin{equation}
    \Sigma^s(g)\psi \doteq \lim_{j\rightarrow \infty}\Sigma^s(g)\psi_j.
\end{equation}
Prop. \ref{prop:sectional_rep} shows this representation is smooth since it maps smooth sections to smooth sections. We have thus proved the following:
\begin{theorem}
    A unitary $G$-equivariant vector bundle $\pi:E\rightarrow M$ equipped with a $G$-invariant volume form $d\xi$ on $M$ induces a smooth unitary representation on $L^2(E,d\xi)$. In particular, the sectional representations of $\mathrm{IO}^+(3,1)$ on $L^2(\gamma,d\xi)$ and $\mathrm{ISO}^+(3,1)$ on $L^2(\gamma_\pm,\xi)$ are unitary.
\end{theorem}
The massless representations constructed in the standard little group representations are representations on the space of $\Comp$-valued $L^2$ functions over $\mathcal{L}_+$ which is equivalent to the space of sections of the trivial line bundle on $\mathcal{L}_+$. This leads to non-smooth representations when $h\neq 0$. By considering representations on sections of nontrivial bundles, one can obtain smooth massless representations.

To complete this picture and prove that massless unitary irreducible bundle representation of $\mathrm{ISO}^+(3,1)$ are particles, it is necessary to show that the corresponding sectional representations are also irreducible:
\begin{theorem}\label{thm:Hilbert_irrep}
    Suppose $\Sigma$ is an irreducible unitary bundle representation of $\poincare$ on $\pi: E \rightarrow \Lightcone$ such that a spacetime translation by $\mathsf{a} \in \mathbb{R}^4$ acts on a vector $(k,v) \in E_k$ by $\Sigma_{\mathsf{a}}(k,v) = (k,e^{ik_\mu \mathsf{a}^\mu}v)$. Then the induced Hilbert space representation $\Sigma^s$ on $L^2(E,d\xi)$ is irreducible.
\end{theorem}
This is harder to prove than the previous results in this section since it depends on specific properties on the Poincar\'{e} group. Indeed, it not generally true that unitary irreducible bundle representations of an arbitrary Lie group $G$ generate irreducible sectional representations. A counterexample can be found in elementary quantum mechanics. Consider the trivial line bundle $E = S^2 \times \Comp$ over $S^2$ with generic element $(\hat{\boldsymbol{k}},c) \in E$ and with the irreducible $\mathrm{SO}(3)$ action which acts by rotations on $\boldsymbol{k}$ and trivially on the fiber:
\begin{equation}
    R(\hat{\boldsymbol{k}},c) = (R\hat{\boldsymbol{k}},c).
\end{equation}
The sections $\Gamma(E)$ are just $\Comp$-valued functions on $S^2$, and the sectional representation is
\begin{equation}
    (Rf)(\hat{\boldsymbol{k}}) = f(R^{-1}\hat{\boldsymbol{k}}).
\end{equation}
This is the standard action of $\mathrm{SO}(3)$ on functions on $S^2$, and is reducible using spherical harmonics. For example, the constant functions are a one-dimensional subrepresentation. In fact, the sectional representation will always be reducible when $G$ is compact. This is because $L^2(E,d\xi)$ is an infinite-dimensional vector space meanwhile any unitary irreducible representation of a compact group is finite-dimensional (\cite{Johnson1976}, Thm. 3.9). However, this does not apply to non-compact groups, such as $\poincare$.

We first prove a simplified version of Theorem \ref{thm:Hilbert_irrep} assuming a smoothness condition. This allows for a more transparent proof and motivates the measure theoretic techniques used in the general proof. Given a representation of $\poincare$ on $L^2(E,d\xi)$, one can decompose it into irreducible Hilbert subrepresentations $\mathcal{H}_a$:
\begin{equation}\label{eq:Hilbert_decomp}
    L^2(E,d\xi) = \bigoplus_a \mathcal{H}_a.
\end{equation}
It is not clear \emph{a priori} that any $\mathcal{H}_a$ contains a smooth section other than the zero section; that is, there is the possibility that every nonzero smooth section is a linear combination of non-smooth elements from the $\mathcal{H}_a$. We call the representation normal if at least one of the $\mathcal{H}_a$ contains a smooth section which is not identically zero.

\begin{theorem}\label{thm:simplified_irrep}
    Theorem \ref{thm:Hilbert_irrep} holds under the additional hypothesis that the sectional representation $\Sigma^s$ is normal.
\end{theorem}

\begin{proof}
    By assumption, $L^2(E,d\xi)$ has an irreducible subrepresentation $\mathcal{H}$ which contains a smooth section $\psi_0$ which is not the zero section. We will prove that $L^2(E,d\xi)=\mathcal{H}$. Consider the vector space $\mathcal{V}$ generated by the orbit of $\psi_0$:
    \begin{equation}\label{eq:orbit_space}
        \mathcal{V} = \{c[\Sigma^s(\Lambda)\psi_0](k) | c \in \Comp, \, \Lambda \in \mathrm{ISO}^+(3,1) \}.
    \end{equation}
    $\mathcal{V} \seq \mathcal{H}$ since $\mathcal{H}$ is invariant under the $\mathrm{ISO}^+(3,1)$ action. The closure of $\mathcal{V}$ in $L^2(E,\xi)$, denoted by $\bar{\mathcal{V}}$, is a closed subspace of $L^2(E,\xi)$ and thus a Hilbert space. Furthermore, $\bar{\mathcal{V}} \seq \mathcal{H}$ since $\mathcal{H}$ is complete. $\bar{\mathcal{V}}$ is also a representation of $\mathrm{ISO}^+(3,1)$ since if $c_j \Sigma(\Lambda_j) \psi_0 \rightarrow \bar{\psi}_0$ and $L \in \mathrm{ISO}^+(3,1)$, then $c_j \Sigma(L \Lambda_j)\psi_0 \rightarrow \Sigma(L)\bar{\psi}_0$ because the action $\Sigma$ is continuous. Since $\mathcal{H}$ is irreducible by assumption, $\mathcal{H}=\bar{\mathcal{V}}$. Decompose $ L^2(E,d\xi)$ as
    \begin{equation}
        L^2(E,d\xi) = \bar{\mathcal{V}} \oplus (\bar{\mathcal{V}})^\perp
    \end{equation}
    where $(\bar{\mathcal{V}})^\perp$ is the orthogonal compliment of $\bar{\mathcal{V}}$. The proof is complete if we show that $(\bar{\mathcal{V}})^\perp=\{0\}$. If $\psi^\perp(k) \in (\bar{\mathcal{V}})^\perp$, then
    \begin{equation}
        \langle \psi^\perp, \Sigma^s(\Lambda)\psi_0 \rangle = \int \frac{\psi^\perp(\boldsymbol{k'})^* \cdot [(\Sigma^s(\Lambda)\psi_0](\boldsymbol{k'})}{|\boldsymbol{k'}|} d\boldsymbol{k'}= 0
    \end{equation}
    for every $\Lambda \in \mathrm{ISO}^+(3,1)$. As $\psi_0(\boldsymbol{k})$ is smooth and not the zero section, there is some small ball $U \in \mathcal{L}_+$ containing a point $\boldsymbol{k}_0$ on which $\psi_0$ is nonvanishing. For each fixed $\boldsymbol{k} \in \Lightcone$, choose a homogeneous Lorentz transformation $L_{\boldsymbol{k}} \in \mathrm{SO}^+(3,1)$ such that $L_{\boldsymbol{k}}\boldsymbol{k}_0 = \boldsymbol{k}$. Note that no continuous choice of $L_{\boldsymbol{k}}$ exists, but this is irrelevant to the current argument. Consider the action of $\mathsf{a} \circ L_{\boldsymbol{k}}$ where $\mathsf{a}$ denotes an arbitrary spacetime translation in only the spatial dimensions $\mathsf{a}=(0,\boldsymbol{\mathsf{a}})$:
    \begin{align}\label{eq:F_k_defn}
    \begin{split}
        \mathcal{F}_{\boldsymbol{k}}(\boldsymbol{\mathsf{a}}) &\doteq \frac{1}{(2\pi)^3} \langle \psi^\perp, \Sigma^s(\mathsf{a} \circ L_{\boldsymbol{k}})\psi_0 \rangle \\
        &=\frac{1}{(2\pi)^3}\int e^{i \boldsymbol{k'} \cdot \boldsymbol{\mathsf{a}}}\frac{\psi^\perp(\boldsymbol{k'})^* \cdot [\Sigma^s(L_{\boldsymbol{k}})\psi_0](\boldsymbol{k'})}{|\boldsymbol{k'}|}d\boldsymbol{k}' \\
        &= 0.
    \end{split}
    \end{align}
    $\mathcal{F}_{\boldsymbol{k}}(\boldsymbol{\mathsf{a}})$ is the 3D Fourier transform of the function
    \begin{equation}\label{eq:f_k_defn}
        F_{\boldsymbol{k}}(\boldsymbol{k}') \doteq \frac{\psi_\perp(\boldsymbol{k'})^* \cdot \psi_{\boldsymbol{k}}(\boldsymbol{k'})}{|\boldsymbol{k'}|}.
    \end{equation}
    where
    \begin{equation}\label{eq:psi_k_defn}
        \psi_{\boldsymbol{k}} \doteq \Sigma^s(L_{\boldsymbol{k}})\psi_0.
    \end{equation}
    Since $\mathcal{F}_{\boldsymbol{k}}(\boldsymbol{a})=0$ for all $\boldsymbol{a}\in \mathbb{R}^3$, and since the Fourier transform is an isomorphism, $F_{\boldsymbol{k}}(\boldsymbol{k}') = 0$ almost everywhere. $E$ is irreducible and therefore a rank-$1$ vector bundle by Theorem \ref{thm:bundle_decomp}. Thus, if $F_{\boldsymbol{k}}(\boldsymbol{k}')=0$ for some $\boldsymbol{k}'$, either $\psi^\perp(\boldsymbol{k}') = 0$ or $\psi_{\boldsymbol{k}}(\boldsymbol{k'})=0$. By construction, $\psi_{\boldsymbol{k}}(\boldsymbol{k'})$ is nonzero for $\boldsymbol{k}'$ in a neighborhood of $\boldsymbol{k}$, and so $\psi^\perp$ must vanish almost everywhere in that neighborhood. By varying $\boldsymbol{k}$, we obtain that $\psi^\perp$ vanishes almost everywhere, and in particular is the zero section of $L^2(E,d\xi)$ since  $L^2$ space is defined up to equivalence almost everywhere. Thus, $(\bar{\mathcal{V}})^\perp = \{0\}$ and $L^2(E,d\xi)=\bar{\mathcal{V}}=\mathcal{H}$ is irreducible. 
\end{proof}

While the assumption that the sectional representation is normal in Theorem \ref{thm:simplified_irrep} appears reasonable, it is known that mathematical pathologies are common in quantum mechanics \cite{Folland2008,Hall2013}. It is thus important to show that the theorem holds in the absence of this smoothness condition.

Without assuming the representation is normal, $\psi_0$ cannot necessarily be chosen to be smooth. However, $L^2$ functions are, by definition, measurable \cite{Stein3_2005}, and we will use measurability as a substitute for smoothness. We will borrow standard results from measure theory and the theory of Hilbert spaces; references for these topics are \cite{Rudin1987,Mattila1999,Stein3_2005}. Let $\xi$ be the measure on the lightcone $\mathcal{L} \cong \pspace$ associated to the Lorentz invariant volume form $d\xi = \frac{d \boldsymbol{k}}{|\boldsymbol{k}|}$. That is, if $A \seq \pspace$ is (Lebesgue) measurable, then
\begin{equation}\label{eq:measure_def}
    \xi(A) = \int_A d\xi.
\end{equation}
$\xi$ can be extended to a measure on all of $\mathbb{R}^3$ by assigning $\xi(\{0\}) = 0$, which is consistent with Eq. (\ref{eq:measure_def}) since 
\begin{equation}
    \xi(\{0\}) = \lim_{r\rightarrow 0}\xi(B_r(\boldsymbol{0})) = \lim_{r\rightarrow 0}\int_{B_r(\boldsymbol{0})} \frac{d \boldsymbol{k}}{|\boldsymbol{k}|} = 0.
\end{equation}
Note that if $\lambda$ denotes the standard Lebesgue measure on $\mathbb{R}^3$, then $\xi(A)=0$ if and only if $\lambda(A) = 0$, and we can refer unambiguously to such sets as measure 0.

We will use a technical result of Buczolich \cite{Buczolich1992}, for which we recall a few definitions. $\boldsymbol{k} \in \mathbb{R}^3$ is a point of Lebesgue density of the measurable set $A\seq \mathbb{R}^3$ if
\begin{equation}\label{eq:Lebesgue_density}
    \lim_{r\rightarrow 0} \frac{\lambda\big(B_r(\boldsymbol{k}) \cap A\big)}{\lambda\big(B_r(\boldsymbol{k})\big)} = 1.
\end{equation}
A function $f:A\rightarrow B$ between two subsets of $\mathbb{R}^3$ is bi-Lipschitz if it is invertible and if both $f$ and $f^{-1}$ are Lipschitz.
\begin{lemma}[Buczolich, \cite{Buczolich1992}, Thm. 1]
    Suppose that $A$ and $B$ are measurable subsets of $\mathbb{R}^3$ and $f:A\rightarrow B$ is bi-Lipschitz. Then $f$ maps points of Lebesgue density of $A$ into points of Lebesgue density of $B$.
\end{lemma}
\begin{corollary}\label{cor:Lorentz_Lebesgue}
    Suppose $A$ is a measurable subset of $\Lightcone\cong \pspace$ and consider $\Lambda \in \poincare$ as a function on $\Lightcone$. Then $\Lambda$ maps points of Lebesgue density of $A$ to points of Lebesgue density of $\Lambda(A)$.
\end{corollary}

\begin{proof}
    By the lemma, it suffices to show that the action of $\Lambda$ is bi-Lipschitz. Furthermore, one can consider only the cases when $\Lambda$ is a pure translation, rotation, or boost. The first two are obviously bi-Lipschitz. For the latter case, without loss of generality, suppose that $\Lambda$ is a boost by velocity $v\hat{\boldsymbol{k}}_x$. $\Lambda$ acts on $\boldsymbol{k}=(k_x,k_y,k_z)$ by
    \begin{equation}
        \Lambda(\boldsymbol{k}) = \Big(\gamma_{Lor} (k_x- v|\boldsymbol{k}|),k_y,k_z\Big)
    \end{equation}
so the Jacobian of the transformation is
\begin{equation}
    |D\Lambda|(\boldsymbol{k}) = \gamma_{Lor}\Big(1-\frac{vk_x}{|\boldsymbol{k}|}\Big) \leq \gamma_{Lor}(1+|v|).
\end{equation}
Since the differential has a uniform bound for all $\boldsymbol{k}$, $\Lambda$ is Lipschitz. $\Lambda^{-1}$ is a boost by $-v\hat{\boldsymbol{k}}_x$, and thus $|D(\Lambda^{-1})|$ is also bounded above by $\gamma_{Lor}(1+|v|)$, so $\Lambda^{-1}$ is Lipschitz. Thus, $\Lambda$ is bi-Lipschitz.
\end{proof}
With these results in place, we can adapt the proof Theorem \ref{thm:simplified_irrep} to prove Theorem \ref{thm:Hilbert_irrep}.
\begin{proof}[Proof of Theorem \ref{thm:Hilbert_irrep}]
    Let $\mathcal{H}$ be an irreducible subrepresentation of $L^2(E,d\xi)$, and fix some section $\psi_0 \in \mathcal{H}$ which is not identically $0$, and define $\mathcal{V}$ as the vector space generated by its orbit as in Eq. (\ref{eq:orbit_space}). Again let $\bar{\mathcal{V}}$ be the Hilbert space completion of $\mathcal{V}$, $(\bar{\mathcal{V}})^\perp$ be its orthogonal complement, and $\psi^\perp \in (\mathcal{V})^\perp$. $\bar{\mathcal{V}}$ is a subrepresentation of $L^2(E,\xi)$ contained in $\mathcal{H}$, so $\bar{\mathcal{V}} = \mathcal{H}$. The strategy is again to show that $(\bar{\mathcal{V}})^\perp$ contains only the zero section. Fix some $\boldsymbol{k}_0 \in \Lightcone$ which will be specified later. For each $\boldsymbol{k} \in \Lightcone$, choose a homogeneous Lorentz transformation $L_{\boldsymbol{k}}$ such that $L_{\boldsymbol{k}}\boldsymbol{k}_0 = \boldsymbol{k}$. By precisely the same argument as in the proof of Theorem \ref{thm:simplified_irrep}, for each fixed $\boldsymbol{k}$, at almost every $\boldsymbol{k}'$ either
    \begin{align}\label{eq:psi_alternative}
       \psi_\perp(\boldsymbol{k}')=0 \; \text{ or }\;
       \psi_{\boldsymbol{k}}(\boldsymbol{k}')=0
   \end{align} 
    where
    \begin{equation}\label{eq:psi_k_defn_2}
        \psi_{\boldsymbol{k}} \doteq \Sigma^s(L_{\boldsymbol{k}})\psi_0.
    \end{equation}
    Define the subsets
    \begin{align}
        R_0 &\doteq \{\boldsymbol{k}' \in \Lightcone ,\, |\psi_0(\boldsymbol{k}')|\neq 0  \} \label{eq:R_0}\\
        R_{\boldsymbol{k}} &\doteq \{\boldsymbol{k}' \in \Lightcone ,\, |\psi_{\boldsymbol{k}}(\boldsymbol{k}')| \neq 0 \} = L_{\boldsymbol{k}}R_0 \label{eq:R_k}\\
        R_\perp &\doteq \{\boldsymbol{k}' \in \Lightcone ,\, |\psi_\perp(\boldsymbol{k}')|\neq 0  \} , \label{eq:R_perp}
    \end{align}
    $|\psi_0|$,$|\psi_{\boldsymbol{k}}|$ and $|\psi_\perp|$ are $L^2(d\xi)$ functions and thus measurable, so the sets $R_0$, $R_{\boldsymbol{k}}$ and $R_\perp$ are measurable. Eq. (\ref{eq:psi_alternative}) says that $R_{\boldsymbol{k}}$ and $R_\perp$ are almost disjoint for every $\boldsymbol{k}$:
    \begin{equation}\label{eq:Rk_Rperp}
        \lambda(R_\perp \cap R_{\boldsymbol{k}}) = 0.
    \end{equation}
    Let $Q_0$, $Q_{\boldsymbol{k}}$, and $Q_\perp$ be the sets of points of Lebesgue density of $R_0$, $R_{\boldsymbol{k}}$, and $R_\perp$. By the Lebesgue density theorem (\cite{Stein3_2005}, Cor. 1.5), a measurable set and its set of points of Lebesgue density differ at most by a set of measure $0$, so
    \begin{align}
        \lambda(R_0) &= \lambda(Q_0) \label{eq:R_0_density}\\
        \lambda(R_{\boldsymbol{k}}) &= \lambda(Q_{\boldsymbol{k}}) \\
        \lambda(R_\perp) &= \lambda(Q_\perp) \label{eq:Rperp_density}.
    \end{align}
    By assumption, $\psi_0$ is not the zero section in $L^2(E,d\xi)$, so $R_0$ has nonzero measure, and therefore so does $Q_0$. We have not yet specified how $\boldsymbol{k}_0$ was chosen; choose $\boldsymbol{k}_0 \in Q_0$. Then $\boldsymbol{k} \in Q_{\boldsymbol{k}}$ by Eq. (\ref{eq:R_k}) and Cor. \ref{cor:Lorentz_Lebesgue}.

    We next establish that
    \begin{equation}\label{eq:Qperp_Qk}
        Q_\perp \cap Q_{\boldsymbol{k}}=\varnothing
    \end{equation}
    for every $\boldsymbol{k}$. Indeed, if $\boldsymbol{k}' \in Q_\perp \cap Q_{\boldsymbol{k}}$, by the definition of Lebesgue density in Eq. (\ref{eq:Lebesgue_density}), there would exists a $\delta>0$ such that 
    \begin{align}
    \begin{split}
    \frac{1}{2}\xi(B_\delta(\boldsymbol{k}')) &< \xi(B_\delta(\boldsymbol{k}') \cap R_\perp), \\
    \frac{1}{2}\xi(B_\delta(\boldsymbol{k}')) &< \xi(B_\delta(\boldsymbol{k}') \cap R_{\boldsymbol{k}}).
    \end{split}
    \end{align}
    By Eq. (\ref{eq:Rk_Rperp}), the intersection of $B_\delta(\boldsymbol{k}')\cap R_\perp$ and $B_\delta(\boldsymbol{k}')\cap R_{\boldsymbol{k}}$ has measure 0. 
    One would then obtain the contradiction
    \begin{align}
    \begin{split}
        \lambda(B_\delta(\boldsymbol{k}')) &\geq \lambda \big(B_\delta(\boldsymbol{k}') \cap (R_\perp \cup R_{\boldsymbol{k}}) \big) \\
        &= \lambda(B_\delta(\boldsymbol{k}') \cap R_\perp) + \lambda(B_\delta(\boldsymbol{k}') \cap R_{\boldsymbol{k}} ) \\
        &> \lambda(B_\delta(\boldsymbol{k}')),
        \end{split}
    \end{align}
    establishing Eq. (\ref{eq:Qperp_Qk}). Thus, $\boldsymbol{k} \notin Q_\perp$ since $\boldsymbol{k} \in Q_{\boldsymbol{k}}$. This holds for every $\boldsymbol{k} \in \Lightcone$, so $Q_\perp = \varnothing$, and therefore $\lambda(R_\perp) = 0$ by Eq. (\ref{eq:Rperp_density}). Thus, $\psi_\perp$ is only nonzero on a set of measure $0$, and is therefore equal to the zero section in $L^2(E,d\xi)$. Then,
$(\bar{\mathcal{V}})^\perp = \{0 \}$, and
\begin{equation}
    L^2(E,d\xi) = \bar{\mathcal{V}} = \mathcal{H}
\end{equation}
is irreducible.
    
\end{proof}

\section{Applications} \label{sec:applications}
In this section we discuss additional applications of the vector bundle description of photons. In particular, we discuss applications to the spin Chern number of light, the quantization of the electromagnetic field, and the spin-orbital decomposition of photon angular momentum.
%%%%%%%%%%% Geometry of spin Chern number %%%%%%%%%%%%%5
\subsection{Geometry of the spin Chern number}
Theorem \ref{thm:RL_equivariant} shows that the splitting $\gamma = \gamma_+ \oplus \gamma_-$ makes sense geometrically as well as topologically. Another way to see the geometry inherent in $\gamma_\pm$ is to note that these bundles are also induced by the Berry curvature, another geometric structure. In particular, it has been shown that the $R$- and $L$-polarization states diagonalize the Berry curvature \cite{Bliokh2015}.

The geometric nature of $\gamma_\pm$ is important for understanding the so-called spin Chern number of light. The Chern classes $c_j$ are the characteristic classes of complex vector bundles, and as such, are topological invariants. Integrals of the Chern classes define the Chern numbers $C_j$ which are also topological invariants and commonly used in physics. Although the Chern numbers are typically calculated using a connection, which is a geometric quantity, they are true topological invariants and independent of the choice of connection \cite{Tu2017differential}. We have shown that $C_1(\gamma) = 0$ and $C_1(\gamma_\pm) = \mp 2$. In addition to the usual Chern number, a more mysterious quantity has been defined for light, the so-called spin Chern number which is physically related to the quantum spin hall effect \cite{Bliokh2015}. Letting $h_\pm = \pm 1$ denote the  
helicities of $\gamma_\pm$, the spin Chern number of light has been defined as
\begin{equation}
    C_{\mathrm{spin}} = h_+C_1(\gamma_+) + h_-C_1(\gamma_-) = 4.
\end{equation}
However, this definition clearly depends on decomposition $\gamma = \gamma_+ \oplus \gamma_-$ of the total photon bundle. From a topological standpoint, we showed in Theorem \ref{thm:infinte_subbundles} that there are an infinite number of possible splittings $\gamma = \ell_j \oplus \ell_{-j}$. The $\gamma_\pm$ decomposition is only preferred when one considers the geometry of $\gamma$ either via Poincar\'{e} symmetry or the Berry connection. Furthermore, the definition of $C_{\textrm{spin}}$ explicitly involves helicity, which is a geometric quantity. Thus, the spin Chern number is not a purely topological quantity, instead reflecting both topological and geometric properties of the bundle $\gamma$. As such, the spin Chern number of light may not be as robust as the Chern number against perturbations to the underlying Maxwell system, particularly if those perturbations alter the geometric properties of the system.

%%%%%%%%%%% Quantization via projection operators %%%%%%%
\subsection{Quantization via projection operators}
We return to the problem of quantizing the vector potential in the Coulomb gauage. If one uses the expansion (\ref{eq:vector_potential}) for $\boldsymbol{A}$, we argued that the polarization vectors $\boldsymbol{\epsilon}_j(\boldsymbol{k})$ cannot be linear. Since we showed that $\gamma$ is trivial and decomposes into two trivial bundles $\gamma = \tau_1 \oplus \tau_2$, it is possible to take $\boldsymbol{\epsilon}_j$ to be global nonvanishing sections of $\tau_j$, thus giving a smooth global choice of polarization vectors. However, there are issues with using the trivial bundles $\tau_j$ since they are not Lorentz invariant, that is, they are not representations of the Lorentz group. The splitting that makes sense both geometrically and topologically is $\gamma = \gamma_+ \oplus \gamma_-$. Indeed, the helicity labels two different types of photons, and from a technical standpoint $R$- and $L$-photons are different particles, although it is conventional to call them both photons \cite{Maggiore2005}. Thus, the $\boldsymbol{\epsilon}_j(\boldsymbol{k})$ would ideally be chosen to be $R$- and $L$-polarizations. The issue again though is that $\gamma_\pm$ are nontrivial and thus no consistent choice of such $\boldsymbol{\epsilon}_j$ exists. However, since $\gamma_\pm$ are well-defined vector bundles, the splitting $\gamma = \gamma_+ \oplus \gamma_-$ allows a section $A(\boldsymbol{k}) \in \Gamma(\gamma)$ to be uniquely written in terms of sections $A_\pm$ of $\gamma_\pm$:
\begin{equation}
    \boldsymbol{A}(\boldsymbol{k}) = A_+(\boldsymbol{k}) + A_-(\boldsymbol{k}).
\end{equation}
That is, there are unique projections $\mathcal{P}_\pm:\Gamma(\gamma) \rightarrow \Gamma(\gamma_\pm)$. Thus, instead of expanding the vector potential in a basis, we write it in terms of projections:
\begin{equation}
    \boldsymbol{\mathcal{A}}(x) = \int \sum_{\sigma=\pm} \frac{1}{\sqrt{2|\boldsymbol{k}|}} [e^{-ik_\mu x^\mu}A_\sigma(\boldsymbol{k}) +e^{i k_\mu x^\mu}A_\sigma^*(\boldsymbol{k})]\frac{d^3\boldsymbol{k}}{(2\pi)^3}.
\end{equation}
In this form all quantities are well-defined and smooth. We can then quantize the field via the usual scheme \cite{Tong2006}, by promoting the $A_\sigma(\boldsymbol{k})^*$ and $A_\sigma(\boldsymbol{k})$ to creation and annihilation operators:
\begin{align}
    [A_\sigma(\boldsymbol{k}_1),A_\eta(\boldsymbol{k}_2)] = [A^*_\sigma(\boldsymbol{k}_1),A^*_\eta(\boldsymbol{k}_2)]= 0 \\
    [A_\sigma(\boldsymbol{k}_1),A^*_\eta(\boldsymbol{k}_2)] = (2\pi)^3\delta^{\sigma \eta} \delta^{(3)}(\boldsymbol{k}_1 - \boldsymbol{k}_2).
\end{align}
Thus, by using projection operators instead of vector expansion, it is possible to quantize the electromagnetic field and obtain the standard QED theory without invoking discontinuous bases as in the standard approach. 

Note that since the Coulomb gauge itself is not Lorentz invariant, the bundle $A(\boldsymbol{k})$ is equal to the $\gamma$ bundle only in the chosen frame where the Coulomb gauge is imposed. To make the quantization Lorentz invariant, the two-dimensional bundles $A(\boldsymbol{k})$ in other frames need to be determined by the requirement of Lorentz invariance.

%%%%%%%%%%% Spin and oribtal angular momentum %%%%%%%%
\subsection{Spin and orbital angular momenta of light}\label{subsec:SAM_OAM}
The last application of our vector bundle methods is to an extended debate about the possibility of splitting photon angular momentum into spin and orbital parts \cite{Akhiezer1965, VanEnk1994,Bliokh2010, Bialynicki-Birula2011, Leader2013, Bliokh2015, Leader2016, Leader2019}:
\begin{equation}
    \boldsymbol{J} = \boldsymbol{J}_s + \boldsymbol{J}_o.
\end{equation}
Note that while photons are massless and thus technically possess helicity not spin, we use the established term spin angular momentum. We work in the sectional representation of the Poincar\'{e} group so that the angular momentum $\boldsymbol{J}$ is represented as an operator on the vector space $\Gamma(\gamma)$. We will show that $\boldsymbol{J}_s$ and $\boldsymbol{J}_o$ do not satisfy $\mathfrak{so}(3)$ commutation relations, and thus cannot properly be considered angular momentum operators. These nonstandard commutation relations have been found by others, although there is no consensus on the implications of them \cite{VanEnk1994, Bliokh2010, Leader2019}. We will use the vector bundle formalism to help explain the meaning of these peculiar commutation relations.

The angular momentum operator $\boldsymbol{J}$ for a section of the photon bundle is determined according to the fact that $\boldsymbol{E}(\boldsymbol{k})$ transforms as a 3-vector when $\boldsymbol{x}$ rotates in $\mathbb{R}^3$, as stated in Theorem \ref{thm:gamma_equivariant}. Direct calculation shows
\begin{align} \label{eq:J=S+L}
    \boldsymbol{J} &= \boldsymbol{S} + \boldsymbol{L}\\
     S_a &\doteq  -i \epsilon_{abc}\\
     \boldsymbol{L} &\doteq - i(\boldsymbol{k}\times \partial_{\boldsymbol{k}})
\end{align}
where $\boldsymbol{S}$ is a 3-vector of rank two tensors. Here, $\boldsymbol{J}$ and $\boldsymbol{L}$ resemble expressions of the total angular momentum and orbital angular momentum for a massive particle, and $\boldsymbol{S}$ assumes the form of the spin-1 operator \cite{Bliokh2010}. If $\boldsymbol{J}$ and $\boldsymbol{S}$ were well-defined operators on the vector space $\Gamma(\gamma)$, then Eq.\,(\ref{eq:J=S+L}) would furnish a split of the angular momentum into spin and orbital parts. The advantage of this decomposition is that $\boldsymbol{L}$ and $\boldsymbol{S}$ satisfy $\mathfrak{so}(3)$ commutation relations, and thus generate rotations as angular momentum operators should. The fundamental issue is that they are ill-defined for photons. The operators should act on sections $\boldsymbol{E}(\boldsymbol{k})$ of the photon bundle where $\boldsymbol{E}$ is embedded in $\Comp^3$. However, $\boldsymbol{L}$ and $\boldsymbol{S}$ generally give $\boldsymbol{E}(\boldsymbol{k})$ a nonzero component in the $\boldsymbol{k}$ direction, violating the transversality condition imposed by Gauss's law. That is to say, $\boldsymbol{L}(\boldsymbol{E})$ and $\boldsymbol{S}(\boldsymbol{E})$ are not sections of $\gamma$. To see the origin of this issue, we start from the derivation of $\boldsymbol{J}$ for a representation of $\mathrm{SO}(3)$ on 3D vector fields $\boldsymbol{\psi}(\boldsymbol{k})$ subject to no constraints. The action of a rotation $R \in \mathrm{SO}(3)$ is given by
\begin{equation}\label{eq:composite_action}
    \big[R \boldsymbol{\psi}\big](\boldsymbol{k}) = R\big[ \boldsymbol{\psi}(R^{-1}\boldsymbol{k})] \big] \doteq (F_{R} \circ \tilde{F}_{R^{-1}}\boldsymbol{\psi})(\boldsymbol{k})
\end{equation}
where $F$ and $\tilde{F}$ are operators on vectors and functions, respectively:
\begin{align}
    F_{R} \boldsymbol{v} &= R\boldsymbol{v} \\
    (\tilde{F}_{R}f)(\boldsymbol{k}) &= f(R\boldsymbol{k}).
\end{align}
Let $R_a(\theta)$ be a rotation by $\theta$ about the $a$-axis. The angular momentum operators $J_a$ are defined by the infinitesimal action of these rotations, and by the chain rule, split into two parts:
\begin{align}
    [J_a\boldsymbol{\psi}](\boldsymbol{k}) &= \frac{d}{d\theta}\Big|_{\theta=0} [R_a(\theta)\boldsymbol{\psi}](\boldsymbol{k}) \\
    &= \Big(\frac{d}{d\theta}\Big|_{\theta=0} F_{R_a(\theta)}\Big) \boldsymbol{\psi}(\boldsymbol{k}) + \Big(\frac{d}{d\theta}\Big|_{\theta=0} \tilde{F}_{R_a^{-1}(\theta)}\Big) \boldsymbol{\psi} (\boldsymbol{k}) \\
    &\doteq S_a\boldsymbol{\psi}(\boldsymbol{k}) + L_a\boldsymbol{\psi}(\boldsymbol{k}). \label{eq:JSL_section_split}
\end{align}

We see that the spin angular momentum, the $S_a\boldsymbol{\psi}(\boldsymbol{k})$ term on right-hand-side of Eq.\,(\ref{eq:JSL_section_split}), is associated with internal rotations, which only change the direction of $\boldsymbol{\psi}$ without acting on the momentum. It is inherently associated with a finite-dimensional representation as it arises from an action of $\mathrm{SO}(3)$ on $\mathbb{R}^3$. On the other hand, orbital angular momentum, the $L_a\boldsymbol{\psi}(\boldsymbol{k})$ term, is associated with an action on functions, and is thus associated with an infinite-dimensional representation of $\mathrm{SO}(3)$. 
The issue in applying this to the photon bundle is that due to the transversality condition, neither $F$ nor $\tilde F$ are well-defined as operators on sections of $\gamma$; only the composite action in Eq.\,(\ref{eq:composite_action}) is well-defined. That is, if $\boldsymbol{E}(\boldsymbol{k}) \in \Gamma(\gamma)$, neither $R[\boldsymbol{E}(\boldsymbol{k})]$ nor $\boldsymbol{E}(R\boldsymbol{k})$ are generally in $\Gamma(\gamma)$. As such, the split of angular momentum into spin and orbital parts as suggested by Eq.\,(\ref{eq:JSL_section_split}) does not apply for $\boldsymbol{\psi} \in \Gamma(\gamma)$. However, Eq.\,(\ref{eq:JSL_section_split}) itself is still valid for $\boldsymbol{\psi} \in \Gamma(\gamma)$, when the right-hand-side is viewed as a single operator.

The failure of splitting angular momentum into spin and orbital parts for photons, however, motivates another potential splitting. In the massive case, we saw that the spin angular momentum is associated with the symmetry group which does not change the momentum $\boldsymbol{k}$. This is precisely how the little group is defined---it is the subset of the Poincar\'{e} group which preserves the momentum. In the massless case, this produces the helicity operator $\chi = \boldsymbol{J}\cdot \boldsymbol{\hat{k}}$. The helicity is the $\boldsymbol{\hat{k}}$ component of the angular momentum $\boldsymbol{J}$, so it appears reasonable to define
\begin{align}
\boldsymbol{J}_s &= (\boldsymbol{J} \cdot \boldsymbol{\hat{k}}) \boldsymbol{\hat{k}} = \chi\boldsymbol{\hat{k}},  \\ \boldsymbol{J}_o &= \boldsymbol{J}_\perp = \boldsymbol{J} - \chi \boldsymbol{\hat{k}}.  
\end{align}
These are indeed well-defined vector operators since 
\begin{align}
    [J_a, J_{s,b}] = i\epsilon_{abc}J_{s,c}, \\
    [J_a, J_{o,b}] = i\epsilon_{abc}J_{o,c}
\end{align}
as one can show from Eqs.\,(\ref{eq:comm_1})-(\ref{eq:comm_4}). In the literature, $\boldsymbol{J}_s$ and $\boldsymbol{J}_o$ are referred to as ``spin angular momentum'' and ``orbital angular momentum'', respectively. This splitting of the photon angular momentum has been proposed based on a number of different arguments \cite{VanEnk1994,Bliokh2010,Bialynicki-Birula2011}. The issue with this splitting is that the $\boldsymbol{J}_s$ and $\boldsymbol{J}_o$ satisfy the peculiar commutation relations
\begin{align}
    [J_{s,a},J_{s,b}]&=0 \label{eq:Lie_1},\\ 
    [J_{o,a},J_{s,b}]&= i \epsilon_{abc}J_{s,c} \label{eq:Lie_2},\\
    [J_{o,a},J_{o,b}]&= i \epsilon_{abc}(J_{o,c}-J_{s,c}) \label{eq:Lie_3}.
\end{align}
In particular, these are not $\mathfrak{so}(3)$ commutation relations, meaning they do not generate rotations. Thus, $\boldsymbol{J}_s$ and $\boldsymbol{J}_o$ cannot be considered angular momenta in the usual sense. This conclusion was also reached by van Enk and Nienhuis \cite{VanEnk1994} and Leader and Lorc\'{e} \cite{Leader2019}. However, if these operators do not generate rotations, what do they generate? One can explicitly check from the commutation relations that the $J_{s,i}$ and $J_{o,i}$ satisfy the Jacobi identity, and therefore form a well-defined Lie algebra $\mathfrak{g}$.

We examine first the fact that the $J_{s,i}$ commute. There is a notable difference between the helicity $\chi = \boldsymbol{J} \cdot \hat{\boldsymbol{k}}$ and the operator $\boldsymbol{J}_{s} = \chi \hat{\boldsymbol{k}}$. Per Theorem \ref{thm:sectional_little_group_rep}, the helicity is associated with the little group action of $\mathrm{SO}(2)$ on $\Gamma(\gamma)$:
\begin{equation}
 \theta \boldsymbol{E}(\boldsymbol{k}) = e^{i\chi \theta}\boldsymbol{E}(\boldsymbol{k}) = e^{i\theta}\boldsymbol{E}_+(\boldsymbol{k}) + e^{-i\theta}\boldsymbol{E}_-(\boldsymbol{k}).
\end{equation}
In contrast, the ``spin angular momentum''  $\boldsymbol{J}_s$  is associated with an $\mathbb{R}^3$ action. If $\boldsymbol{v}\in \mathbb{R}^3$, then
\begin{equation}
    \boldsymbol{v} \boldsymbol{E}(\boldsymbol{k}) \doteq e^{i\chi (\hat{\boldsymbol{k}}\cdot \boldsymbol{v})}\boldsymbol{E}(\boldsymbol{k}) = e^{i(\hat{\boldsymbol{k}}\cdot \boldsymbol{v})}\boldsymbol{E}_+(\boldsymbol{k}) + e^{-i(\hat{\boldsymbol{k}}\cdot \boldsymbol{v})}\boldsymbol{E}_-(\boldsymbol{k}).
\end{equation}
Thus, the $J_{s,i}$ are associated with a translational symmetry of $\gamma$, explaining why they form a three-dimensional commuting Lie subalgebra of $\mathfrak{g}$. On the other hand, the $\boldsymbol{J}_o$ do not form a Lie subalgebra as seen by Eq.\,(\ref{eq:Lie_3}). Thus, $\boldsymbol{J}_o$ is not associated with any symmetry of the photon bundle. Since neither $\boldsymbol{J}_o$ nor $\boldsymbol{J}_s$ are related to rotational symmetries of the Maxwell system, they do not achieve a spin-orbital decomposition of the angular momentum. This traces back to the fact that photons are massless, and thus have helicity rather than spin. Indeed, it can be said that the ``spin angular momentum'' $\boldsymbol{J}_s$ is neither spin nor angular momentum. Instead, we have seen that $\boldsymbol{J}_s$ is associated with a helicity-induced translational symmetry of the photon system. We emphasize that even though $\boldsymbol{J}_s$ and $\boldsymbol{J}_o$ are not truly angular momenta, the splitting is well-defined and has proved useful in experimental optics; see Bliokh et al. \cite{Bliokh2014} and references therein. Thus, a proper understanding of the operators $\boldsymbol{J}_s$ and $\boldsymbol{J}_o$ is of both theoretical and experimental import in various applications \cite{ruiz2015first,ruiz2015lagrangian,ruiz2017extending,ruiz2017geometric,oancea2020gravitational,Fu2023}.

We note that discontinuous polarization bases have appeared in some treatments of the photon angular momentum. Bliokh \emph{et al.} \cite{Bliokh2010} write the electric field in the helicity basis
\begin{equation}
\boldsymbol{e}^\pm(\boldsymbol{k})=e^{\pm im\phi}\big(\boldsymbol{e}_\theta(\boldsymbol{k}) \pm i\boldsymbol{e}_\phi(\boldsymbol{k})\big),
\end{equation}
where $\boldsymbol{e}_\theta$ and $\boldsymbol{e}_\phi$ are the usual polar unit vectors and $m$ is an integer. This basis, however, has singularities at the poles, and is thus not globally well-defined. Similarly, I. Bialynicki-Birula and Z. Bialynicki-Birula \cite{Bialynicki-Birula2011} incorrectly assume that global polarization vectors for the $R$- and $L$-polarizations exist. An advantage of using the equivariant bundle formalism to discuss photon angular momentum is that there is no need to invoke discontinuous polarization bases. As shown in Eq.\,(\ref{eq:E=E++E-}), every photon wave function can be uniquely decomposed into the $R$- and $L$-components through the projection operators.

%%%%%%%%%%%%% Conclusion %%%%%%%%%%%%%%%%%%%
\section{Conclusion}

Despite the simplicity of the Maxwell system, it exhibits surprisingly rich and subtle topological behavior. While the total photon bundle is trivial, it has important topologically nontrivial subbundles such as the $R$- and $L$-bundles. This nontrivial topology traces back to the hole in momentum space at $\boldsymbol{k}=0$, accounting for the fact that photons are massless and have no rest frame. This nontrivial topology frequently obstructs the smoothness of constructions that work in topologically trivial cases, such as the little group construction on Hilbert space representations and the quantization via expansion in a polarization basis. We showed that vector bundle methods can be used to avoid these continuity issues. In particular, equivariant vector bundles have precisely the right structure to simultaneously study the topology and symmetry of waves. In the present case of photons, this formalism allows for versions of the little group construction and quantization of the electromagnetic field without encountering discontinuities. It also elucidated the symmetry issues that occur in attempts to separate photon angular momentum into spin and orbital parts. The equivariant bundle formalism is very general, and can be applied to any waves with arbitrary symmetry groups. As such, we believe that it could be a useful framework for the general study of topological properties of waves.

\begin{acknowledgments}
This work is supported by U.S. Department of Energy (DE-AC02-09CH11466).
\end{acknowledgments}

%%%%%%%%%%%%%% Bibliography %%%%%%%%%%%%%%%%%%%%

\bibliography{pt}
\end{document}